\documentclass[aps,pra,twocolumn,superscriptaddress]{revtex4-1}

\usepackage[english]{babel}
\usepackage{amsmath}
\usepackage{amsthm}
\usepackage{amssymb}
\usepackage{booktabs}
\usepackage{color}
\usepackage{hyperref}
\usepackage{cases}
\usepackage{multirow}
\hypersetup{
  colorlinks   = true,  
  urlcolor     = blue,  
  linkcolor    = blue,  
  citecolor    = red    
}
\usepackage{natbib}
\usepackage{pgfplots}
\usepackage{subfigure}
\usepackage{url}

\newtheorem{theorem}{Theorem}

\theoremstyle{definition}
\newtheorem{definition}{Definition}

\theoremstyle{remark}

\newtheorem{lemma}[theorem]{Lemma}

\usepackage{scalerel,stackengine}
\stackMath
\newcommand\reallywidehat[1]{%
\savestack{\tmpbox}{\stretchto{%
  \scaleto{%
    \scalerel*[\widthof{\ensuremath{#1}}]{\kern-.6pt\bigwedge\kern-.6pt}%
    {\rule[-\textheight/2]{1ex}{\textheight}}
  }{\textheight}%
}{0.5ex}}%
\stackon[1pt]{#1}{\tmpbox}%
}
\usepackage{comment}
\usepackage{caption}
\usepackage{float}

\begin{document}

\title{Optimization of the Variational Quantum Eigensolver for Quantum Chemistry Applications}

\author{R.J.P.T. \surname{de Keijzer}}
\affiliation{Eindhoven University of Technology, P.~O.~Box 513, 5600 MB Eindhoven, The Netherlands}
\altaffiliation[Corresponding author: ]{r.d.keijzer@tue.nl }

\author{V.E. \surname{Colussi}}
\affiliation{Eindhoven University of Technology, P.~O.~Box 513, 5600 MB Eindhoven, The Netherlands} \affiliation{INO-CNR BEC Center and Dipartimento di Fisica, Università di Trento, 38123 Povo, Italy}

\author{B. \surname{\u{S}kori\'{c}}}
\affiliation{Eindhoven University of Technology, P.~O.~Box 513, 5600 MB Eindhoven, The Netherlands}

\author{S.J.J.M.F. \surname{Kokkelmans}}
\affiliation{Eindhoven University of Technology, P.~O.~Box 513, 5600 MB Eindhoven, The Netherlands}

\date{\today}

\begin{abstract}
This work studies the variational quantum eigensolver algorithm, which is designed to determine the ground state of a quantum mechanical system by combining classical and quantum hardware. Two methods of reducing the number of required qubit manipulations, prone to induce errors, for the variational quantum eigensolver are studied. First, we formally justify the multiple $\mathbb{Z}_2$ symmetry qubit reduction scheme first sketched by Bravyi et al. [arXiv:1701.08213 (2017)]. Second, we show that even in small, but non-trivial systems such as H$_2$, LiH, and H$_2$O, the choice of entangling methods (gate based or native) gives rise to varying rates of convergence to the ground state of the system.  Through both the removal of qubits and the choice of entangler, the demands on the quantum hardware can be reduced. We find that in general, analyzing the VQE problem is complex, where the number of qubits, the method of entangling, and the depth of the search space all interact. In specific cases however, concrete results can be shown, and an entangling method can be recommended over others as it outperforms in terms of difference from the ground state energy.
\end{abstract}

\maketitle

\section{Introduction}
\label{sec:section1}
Presently, quantum computing is in the noisy intermediate-scale quantum (NISQ) era \cite{Preskill_2018}, where the available universal quantum computers cannot outperform their classical counterparts except for a few, specific, well-designed cases \cite{google}.  
However, even in the NISQ era quantum computers can be used in practice. One such application
is the variational quantum eigensolver (VQE) algorithm, first proposed in a paper by Peruzzo and McClean in 2014 \cite{firstmention}. The algorithm is hybrid, e.g. at certain points a quantum processor unit (QPU) is addressed. The VQE algorithm has shown proof of concept for small molecules with several different designs of qubits \cite{Kandala,trapion1,photonics1,majorana1,googlehartree}. The aim of the algorithm is to determine the lowest eigenvalue of a quantum Hamiltonian, which can be equated to finding the ground state energy of a molecule. In recent years VQE has been implemented on
many different qubit systems and has become a highly active
area of research (c.f. Refs.~\cite{overview1,overview2} for recent comprehensive overviews).\\

To illustrate the VQE algorithm, a schematic of the relevant circuit diagram is shown in Fig.~\ref{fig:diagram}. To determine the ground state energy of a Hamiltonian the VQE algorithm maps the fermionic degrees of freedom of molecules onto a set of qubits. These qubits can then be manipulated into a trial state with a certain energy by applying a sequence of entanglement operators $U_{ent}$ and individual rotations $U_{i,j}$. The length of this sequence is the depth $d$ which plays an important role in determining which part of the state space can be addressed by the VQE. The entanglement operator $U_{ent}$ is a multi-qubit gate that ensures entangled states can be reached and thus its choice is critical in exploiting the quantum nature of the method. This operator is necessary for utilizing the computational powers of the multi-qubit structure. By applying a problem dependent Hamiltonian to the trial state, its energy can be measured using a QPU. Based on these measurements the classical part of the computer proposes a new trial state (new rotations) by means of a classical optimization algorithm \cite{penalty,ibmsteps}. Because of the limited hardware of quantum computers, each qubit manipulation has a non-negligible probability of error. Therefore, it is important to minimize the number of necessary manipulations.

\begin{figure}[H]
\includegraphics[scale=0.28]{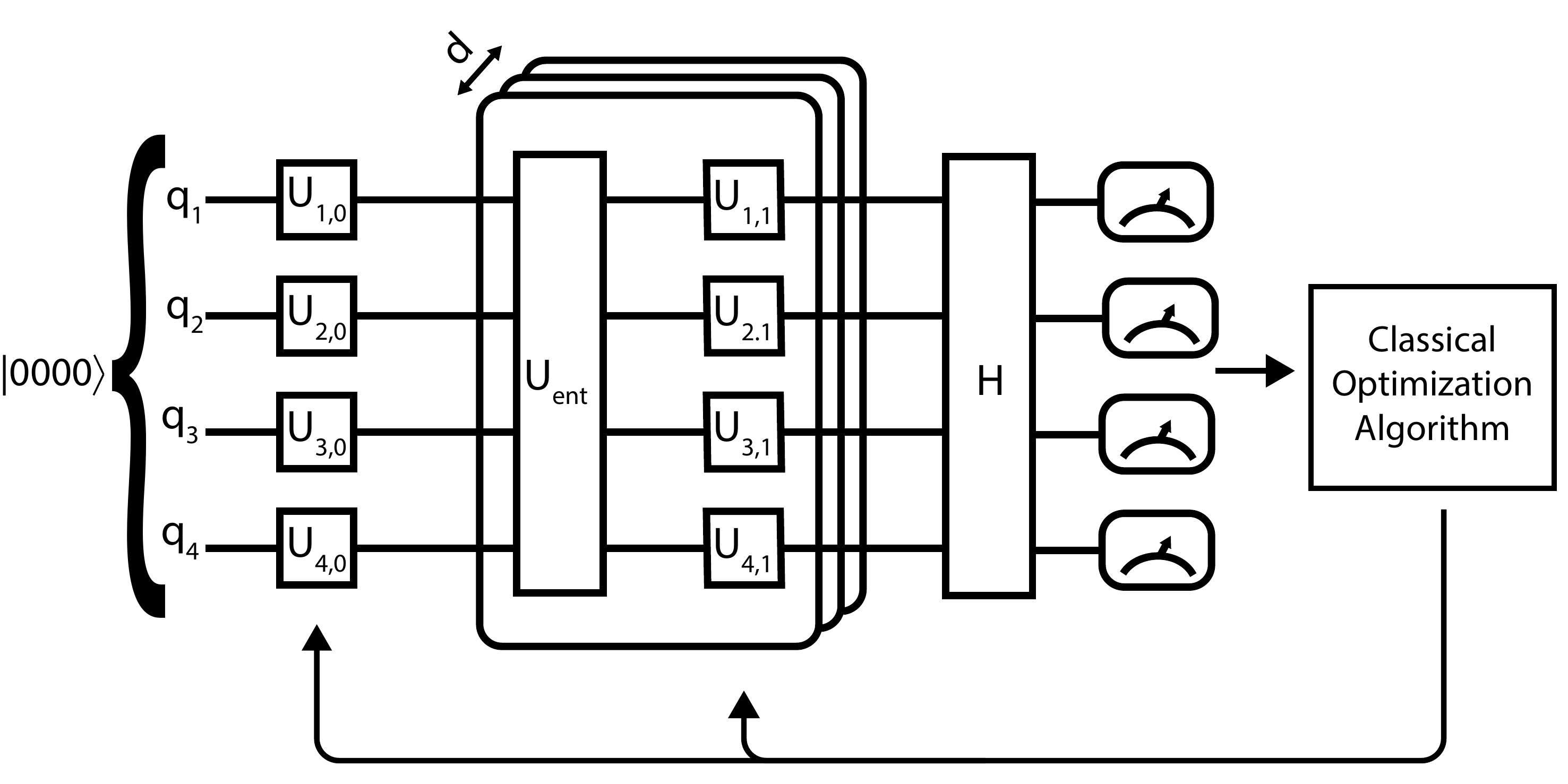} 
\caption{Schematic  circuit  diagram  of  the  VQE  algorithm.  The  qubits  are manipulated by a sequence of entanglement operators $U_{ent}$ and individual rotations $U_{i,j}$ into a prepared trial state.
}
\label{fig:diagram}
\end{figure}

This paper examines two ways of achieving this minimization: (i) by reducing the number of qubits used to describe the problem as proposed by Bravyi et al. \cite{taper} and (ii) by reducing the number of qubit manipulations using entanglement methods native to the physical system. Both of these methods are directed to optimizing the quantum hardware rather than the classical side. \\

The layout of this paper is as follows. In Secs.~\ref{sec:section2} and \ref{sec:section3} the VQE algorithm and the method of simulation are described. Sec.~\ref{sec:section4} gives a constructive proof of the multiple $\mathbb{Z}_2$ qubit reduction symmetries scheme by Bravyi et al. \cite{taper} Sec.~\ref{sec:section5} describes multi-qubit entangling methods used in the VQE analysis. Finally, in Sec.~\ref{sec:section7} we compare and contrast the entanglement methods and their results on small molecule VQE test case problems. 

\section{Quantum measurements and simulation setup}
\label{sec:section2}
The VQE algorithm exploits the variational principle and aims to find a set of parameters minimizing the energy of a quantum system. The states of a system are mapped to qubits and the classical optimization algorithm optimizes the variational calculation. To do so for Hamiltonians of small molecules requires a characterization of their electron spin-orbitals. For this the Hartree-Fock method is applied, which reformulates the Hamiltonian in first-quantization form \cite{chemicalaccuracy}. The Born-Oppenheimer approximation is applied which ensures that the nuclear contribution to the Hamiltonian is constant \cite{chemicalaccuracy}. The Hamiltonian $H_1$ under the Born-Oppenheimer approximation becomes
\begin{equation}
\begin{aligned}
    H_1&=-\sum_i \frac{\nabla^2_{\vec{r}_i}}{2}-\sum_{i,j}\frac{Z_i}{|\vec{R}_i-\vec{r}_j|}\\&+\sum_{i,j>i}\frac{1}{|\vec{r}_i-\vec{r}_j|}+\sum_{i,j>i}\frac{Z_iZ_j}{|\vec{R}_i-\vec{R}_j|},
    \end{aligned}
\end{equation}
where, $\vec{R}_i$ and $Z_i$ are the position and charge of nucleus $i$. $\vec{r}_i$ is the position of electron $i$. In VQE, the Hamiltonian is described in Hartree units \footnote{ Hartree units are also known as atomic units where the reduced constant of Plank $\hbar$, the electron mass $m_e$, the elementary charge $e$ and Coulomb's constant $k_e=1/4\pi\epsilon_0$ are equal to unity. In this system the Bohr radius $a_0=4\pi\epsilon_0\hbar^2/m_e e^2$ is also equal to unity. Hartree units are often used in molecular level calculations \cite{mcquarrie2008quantum}.}. The Hamiltonian is reformulated in second-quantization form \cite{secondquantization} $H_2$, by projecting it onto a finite set of orthogonal spin-orbital modes $\{\phi_1\}_{i=1}^n$, as
\begin{equation}
\label{eq:secondquant}
    H_2=V_{nn}+\sum_{p,q}h_{pq}a^\dagger_p a_q+\frac{1}{2}\sum_{p,q,r,s}h_{pqrs}a^\dagger_p a^\dagger_q a_r a_s,
\end{equation}
where the nuclei-nuclei interaction potential $V_{nn}=\sum_{i,j>i}Z_iZ_j/|\vec{R}_i-\vec{R}_j|$ is constant and the annihilation and creation operators $a_i^\dagger$ and $a_i$ work on a set of orthogonal spin-orbital modes $\{\phi_1\}_{i=1}^n$ determined by the linear combination of atomic orbitals (LCAO) method and which induce the fermionic algebra \cite{mcquarrie2008quantum}. The number of considered orbitals determines the complexity of the problem and with that the number of required qubits. In many practical cases only a select few low energy orbitals are considered while all higher energy atomic orbitals are ignored. The molecular integrals $h_{pq}$ and $h_{pqrs}$ are referred to as the one-electron and two-electron integrals respectively 
\begin{equation}
\begin{aligned}
    h_{pq}&=\int S^*_p(\omega)S_q(\omega) d \omega \\
    &\enskip\times\int_{\mathbb{R}^3} \phi^*_p(\vec{r})\left(\frac{\nabla^{2}_{\vec{r}}}{2}-\sum_i\frac{Z_i}{|\vec{R_i}-\vec{r}|} \right)\phi_q(\vec{r}) d\vec{r},
\end{aligned}
\label{eq:oneelectron}
\end{equation}

\begin{equation}
\begin{aligned}
    h_{pqrs}=&\int\int S^*_p(\omega_1)S_s(\omega_1)S^*_q(\omega_2)S_r(\omega_2)d \omega_1 d\omega_2\\
    &\times\int_{\mathbb{R}^3} \int_{\mathbb{R}^3} \frac{\phi^*_p(\vec{r_1})\phi^*_q(\vec{r_2})\phi_r(\vec{r_2})\phi_s(\vec{r_1})}{|\vec{r_1}-\vec{r_2}|}d\vec{r_1}d\vec{r_2},
\end{aligned}
\label{eq:integrals}
\end{equation}
where $S_i(\omega)$ is the spin part of the spin orbitals \footnote{In this work the integrals in Eq.~\ref{eq:oneelectron} and Eq.~\ref{eq:integrals} have been determined using  quantum computation libraries \textit{OpenFermion} \cite{openfermion} and \textit{Psi4} \cite{psi4} in the STO-3G basis.}.\\

Some important concepts in quantum computing are introduced below. These are required to understand the VQE method and are used in proofs that follow.
\begin{definition}{Pauli Matrices}
\begin{equation}
    I:= \begin{pmatrix}
  1 & 0 \\
  0 & 1 \\
 \end{pmatrix},\enskip
 X:= \begin{pmatrix}
  0 & 1 \\
  1 & 0 \\
  \end{pmatrix},\enskip
 Y:= \begin{pmatrix}
  0 & -i \\
  i & 0 \\
 \end{pmatrix},\enskip
  Z:= \begin{pmatrix}
  1 & 0 \\
  0 & -1 \\
 \end{pmatrix}.
\end{equation}
\end{definition}

\begin{definition}{$m$-fold Single Qubit Operator}
\begin{equation}
    \sigma^\rho_j:=I^{\otimes (j-1)}\otimes\rho\otimes I^{\otimes (m-j)},\quad j\in{1,...,m}, \quad \rho\in\{X,Y,Z\}.
\end{equation}
\end{definition}

\noindent These $m$-fold single qubit operator work on a $m$-qubit state but only change the state of one of the qubits.

\begin{definition}{$m$-fold Pauli operator space $P_m$}
\begin{equation}
    P_m:=\{\pm 1,\pm i\}\times\{I,X,Y,Z\}^{\otimes m}.
\end{equation}
\end{definition}

\noindent $P_m$ forms a group under operator multiplication since the Pauli operators $P_1$ form a group. Note that since the Pauli matrices are a complete set for the Hilbert space of complex $2\times2$ matrices, $P_m$ is a complete set for the Hilbert space of complex $2^m\times2^m$ matrices\\

The Jordan-Wigner transformation can be used to map the creation and annihilation operators of Eq.~\eqref{eq:secondquant} to combinations of Pauli matrices \cite{transforms}. The transformation is given by 
\begin{equation}
    a_j=I^{\otimes j-1}\otimes\sigma^+\otimes Z^{\otimes m-j},
\end{equation}
\begin{equation}
    a^\dagger_j=I^{\otimes j-1}\otimes\sigma^-\otimes Z^{\otimes m-j},
\end{equation}

\noindent where $m$ is the number of qubits. The $\sigma^\pm$ operators are defined as
\begin{equation}
    \sigma^\pm:=\frac{X\pm iY}{2}.
\end{equation}

It can be shown that the transformed operators $a_j$ and $a^\dagger_j$ obey the canonical commutation relations for fermions. Since this transformation is local the electron configuration can be determined from the qubit state. Since the number of required qubits is a key factor determining the required effort for calculation on both physical systems and in simulations it is desirable to reduce the number of qubits where possible. The Hamiltonian as constructed in the Jordan-Wigner transformation commutes with the number spin up and number spin down operators $N_{\uparrow}$ and $N_{\downarrow}$ so we have
\begin{equation}
\label{eq:project}
    [H,N_{\uparrow}]=[H,N_{\downarrow}]=0.
\end{equation}
This symmetry can be used to taper off 2 qubits based respectively on particle and spin number as described by Bravyi et al. \cite{taper}. This effectively projects the Hamiltonian on those states describing the right number of electrons and spin. This method can be extended to taper off more qubits if additional symmetries are present in the Hamiltonian, which is the subject of Sec.~\ref{sec:section4}.\\

Because $P_m$ is a complete set for the Hilbert space of complex $2^m\times2^m$ matrices, any $m$-qubit Hamiltonian $H$ can be written as a linear combination of elements of the $m$-fold Pauli operators in $P_m$ as
\begin{equation}
\begin{aligned}
    H&=\sum_k h_k \boldsymbol{\sigma}^k,\quad \boldsymbol{\sigma}^k=\sigma^k_1\otimes\sigma^k_2\otimes...\otimes\sigma^k_m\in P_m, \quad\\
    h_k&=\langle H,\boldsymbol{\sigma}^k\rangle_{\mathcal{F}}
    =\sum_{i,j}\overline{H^T_{ij}}\boldsymbol{\sigma}^k_{ij}=tr(H^\dagger\boldsymbol{\sigma}^k). 
    \end{aligned}
    \label{eq:frobenius}
\end{equation}

\noindent This is a finite dimensional Fourier decomposition with the Frobenius inner product $\langle\cdot,\cdot\rangle_\mathcal{F}$ on matrix representations of the operators. Note that in Eq.~\eqref{eq:frobenius} $\sigma_{ij}^k$ is a matrix element, not a Pauli matrix working on a qubit. Any rotation on a single qubit can be written as
\begin{equation}
   R(\alpha,\beta,\gamma,\delta)= e^{i\alpha}Z_\beta X_\gamma Z_\delta,
\end{equation}

\noindent where $\alpha,\beta,\gamma,\delta \in [0,2\pi]$. The operators $X_\phi$ and $Z_\phi$ denote a rotation of an angle $\phi$ around the given axis and are given by
\begin{equation}
\label{eq:rotation1}
\begin{aligned}
    &X_\phi=\cos(\phi/2)I+\sin(\phi/2)X, \\ &Y_\phi=\cos(\phi/2)I-i\sin(\phi/2)Y,\\
    &Z_\phi=\cos(\phi/2)I+\sin(\phi/2)Z.
    \end{aligned}
\end{equation}
\newpage
\section{Trial State Initialization and Measurement}
\label{sec:section3}
In order to initialize a trial state the individual qubits have to be put in a requested state. Following Kandala et al. \cite{Kandala}, we assume control over the rotations of the individual qubits and let the qubits entangle by a passive \textit{always-on} interaction described by some drift Hamiltonian $H_0$. As long as the rotations can be performed in a relatively small time frame compared to the time between two rotations $t_0$ we can model the evolution of the system as a sequence of rotations and entanglement operations $U_{ent}=e^{iH_0t_0}$. This method is called the hardware efficient ansatz \cite{ansatzes} and is especially relevant in NISQ machines where there will likely be a passive interaction \cite{Preskill_2018}. If such an interaction is not present and the qubit system possesses some form of control in entanglement operations then there is more freedom in choosing what $U_{ent}$ should look like. The depth $d$ of a state preparation is defined as the length of the  sequence of rotations and entanglement operations, or equivalently the number of entanglement operations applied. Each state on the Bloch sphere of a single qubit can be reached with a $ZXZ$-rotation. Such a rotation on a qubit $q$ at a depth $i$ can be written as
\begin{equation}
    U_{q,i}(\vec{\theta})=Z_{\theta_1^{q,i}}X_{\theta_2^{q,i}}Z_{\theta_3^{q,i}},
\end{equation}

\noindent where $\vec{\theta}$ has three elements for every qubit and depth pair. In total $3d+3$ rotations would be performed on every qubit. However, in this paper the initial state will always be the vacuum state. The first $Z$-rotation can therefore be omitted. Thus the parameter vector $\vec{\theta}$ is an element of the search space $[0,2\pi]^{\otimes D}$, where $D=(3d+2)m$. To reach the trial state described by the parameter vector $\vec{\theta}$ the initial state $|\psi_{init}\rangle$ is rotated and entangled. The entanglement operator is $U_{ent}$. The prepared trial state after all rotations and entanglements will be
\begin{equation}
\label{equationstatepreperation}
\begin{aligned}
    |\Psi(\vec{\theta})\rangle=&\left(\prod^m_{q=1}U_{q,d}(\vec{\theta})\times U_{ent}\right)\times\left(\prod^m_{q=1}U_{q,d-1}(\vec{\theta})\times U_{ent}\right)\times...\\
    &\times\left(\prod^m_{q=1}U_{q,0}(\vec{\theta})\right)|\psi_{init}\rangle.
\end{aligned}
\end{equation}

With the trial state prepared, the expectation value of $H$ can be measured on a quantum computer as 
\begin{equation}
\label{hamiltonianeval}
\begin{aligned}
    \langle H \rangle_{\vec{\theta}}=&\langle\Psi(\vec{\theta})|H|\Psi(\vec{\theta})\rangle=\sum_k h_k\langle\Psi(\vec{\theta})|\boldsymbol{\sigma}^k|\Psi(\vec{\theta})\rangle\\
    &= \sum_k h_k \langle\Psi(\vec{\theta})|\sigma_1^k\otimes \sigma_2^k\otimes...\otimes\sigma_m^k|\Psi(\vec{\theta})\rangle.
    \end{aligned}
\end{equation}
In quantum computers these expectation values require sufficient measurements. In simulation the expectation values require an inner product of matrices and vectors. 
\newpage

\section{$\mathbb{Z}_2$ symmetry qubit reduction scheme}
\label{sec:section4}
This section formally justifies the qubit removal process for multiple $\mathbb{Z}_2$ symmetries as described by Bravyi et al. \cite{taper}. In this process the $\mathbb{Z}_2$ symmetries in Hamiltonians \cite{taper}, often originating from geometric symmetries in the molecule, are exploited to reduce the number of qubits necessary to describe the Hamiltonian. These $\mathbb{Z}_2$ symmetries include the spin and electron number symmetries and have recently been linked to more physically intuitive point group symmetries \cite{pointgroup}.  The reductions will substantially lower computation power and time. The qubit reductions based on $\mathbb{Z}_2$ symmetries related to spin and electron number have found substantial use in recent works \cite{overview1,bravyi1,bravyi2} and have been included in quantum computing libraries, e.g. \textit{OpenFermion} \cite{openfermion}. Implementations of the additional  $\mathbb{Z}_2$ symmetries discussed in this work have are limited in number, most likely because there are not many or even no such symmetries present in the small systems that can analyzed in classical simulations. However, in larger systems they can lead to substantial reduction in hardware requirement \cite{taper}. The main idea of the qubit reduction scheme is to transform a $m$-qubit Hamiltonian in such a way that the Pauli operators in the decomposition will all have either $I$ or $X$ as the last $r\in\mathbb{N}$ factors (e.g. YXXZIX and ZXYZXX). If the $i$-th term of every Pauli operator in the decomposition is $I$ or $X$ then $[H,\sigma^x_i]=0$, so the Hamiltonian $H$ and $\sigma^x_i$ commute and thus share common eigenvectors. One can therefore replace the $i$-th factors of the operators with the eigenvalues $\pm 1$ and thus taper off a qubit \cite{taper}. We ignore the phases of the $m$-fold Pauli operators because these do not influence commutation. Doing so every $m$-fold Pauli operator becomes its own inverse.\\

In the following proof, we place the qubit reduction scheme sketched in Sec.VII of the paper by Bravyi et al. \cite{taper} on a mathematically rigorous level.  We begin first with some definitions and lemmas required to provide a constructive proof of the scheme.  

\begin{definition}{Symmetry Group}

\noindent A symmetry group $S$ of a group $G$ is an abelian subgroup of $G$ such that $-I\not\in S$.
\end{definition}

\begin{definition}{Center and Centralizer}

\noindent Let $G$ be a group. The center $Z(G)$ of $G$ is the set of elements commuting with every element of $G$.
\begin{equation}
    Z(G)=\{z\in G \quad | \quad zg=gz \quad \forall g\in G\}.
\end{equation}
The centralizer $C_G(A)$ of $A\subseteq G$ is the set of all elements in $G$ which commute with every element of $A$.
\begin{equation}
    C_G(A)=\{g\in G \quad | \quad ag=ga \quad \forall a\in A\}.
\end{equation}
\end{definition}

\begin{lemma}{Symmetry group commuting with $2^m\times 2^m$ Hamiltonian}\\
\noindent Let $S\subseteq P_m$ be a symmetry group. $S$ commutes with a $2^m\times 2^m$ Hamiltonian $H=\sum_k h_k \boldsymbol{\sigma}^k$ if $S$ commutes with every Pauli operator in the decomposition of H. Thus, let $H = \sum_k h_k \boldsymbol{\sigma}^k$.
If $\forall k$   $S$ commutes with $\boldsymbol{\sigma}^k$ then  $S$ commutes with $H$.
\end{lemma}

\begin{proof}
$sH=s\sum_k h_k \boldsymbol{\sigma}^k=\sum_k h_k s \boldsymbol{\sigma}^k=\sum_k h_k \boldsymbol{\sigma}^k s=Hs$
\end{proof}

\begin{lemma}{Redefining of generators}\\
Let the generators $\tau_1,\tau_2,...,\tau_r$ generate the abelian group $S=\langle\tau_1,\tau_2,...,\tau_r\rangle\subset P_m$ and let $i,j\in\{1,....r\}$ with $i\neq j$. Then $S=\langle\tau_1,\tau_2,...,\tau_i,...,\tau_j,...,\tau_r\rangle=\langle\tau_1,\tau_2,...,\tau_i,...,\tau_{j-1},\tau_i\tau_j,\tau_{j+1},...,\tau_r\rangle$
\end{lemma}

\begin{proof}
Since $\tau_i,\tau_j\in S$ it holds that $\tau_i\tau_j\in S$. Therefore, $\langle\tau_1,\tau_2,...,\tau_i,...,\tau_{j-1},\tau_i\tau_j,\tau_{j+1},...,\tau_r\rangle\subseteq\langle\tau_1,\tau_2,...,\tau_i,...,\tau_j,...,\tau_r\rangle$. Now if $\tau_i\tau_j,\tau_i \in S$ then $\tau_i\tau_i\tau_j=\tau_j\in S$. Therefore $\langle\tau_1,\tau_2,...,\tau_i,...,\tau_j,...,\tau_r\rangle\subseteq \langle\tau_1,\tau_2,...,\tau_i,...,\tau_{j-1},\tau_i\tau_j,\tau_{j+1},...,\tau_r\rangle$.
\end{proof}

\begin{lemma}{Commutation of product}\\
Let $A,B_1,B_2$ and $B_3$ be operators such that $A$ commutes with $B_1$ and anti-commutes with $B_2$ and $B_3$.  Then $B_1B_2$ and $B_2B_1$ anti-commute with A. $B_2B_3$ and $B_3B_2$ commute with A.
\end{lemma}

\begin{proof}
$AB_1B_2=B_1AB_2=-B_1B_2A$ and $AB_2B_1=-B_2AB_1=-B_2B_1A$. So indeed $B_1B_2$ and $B_2B_1$ anti-commute with A. Furthermore, $AB_2B_3=-B_2AB_3=B_2B_3A$ and $AB_3B_2=-B_3AB_2=B_3B_2A$. So indeed $B_2B_3$ and $B_3B_2$ anti-commute with A. 
\end{proof}

An important result of stabilizer theory \cite{brun2019quantum} is that a symmetry group commuting with a $m$-qubit Hamiltonian will at most require $m$ generators. This is a key fact in the qubit reduction scheme (the exact proof of this is given by Fujii \cite{Nqbitsnec}). A proof sketch would take the following form. Each symmetry group commuting with a $m$-qubit Hamiltonian is a symmetry group of $P_m$. Each symmetry group generator projects a $m$-qubit basis state onto one of its eigenvalues $\pm1$. Therefore, each of the generators divides the Hilbert space of $m$-qubits into two subspaces. The generators share a common set of eigenvectors as they all commute. The generators thus partition the $m$-qubit space into $2^{\mathrm{\#generators}}$ subspaces. Since there are only $2^m$ basis states, $\mathrm{\#generators}\leq m$. 
\\

Now the algorithm for the qubit reduction is given below. For notation purposes the single qubit operators $X,Y$ and $Z$ on qubit $i$ are sometimes denoted as $\sigma_i^X, \sigma_i^Y$ and $\sigma_i^Z$ respectively. The first goal is to find a symmetry group $S$, commuting with $H$, which has a maximal number of generators. First off, every Pauli operator $\boldsymbol{\sigma}^k\in P_m$ will be encoded as a binary row-vector $(a_x|a_z)$ of length $2m$ as follows
\begin{equation}
    \sigma(a_x|a_z)=\sigma(a_{x1},...,a_{xm}|a_{z1},...,a_{zm})=\prod_{i=1}^m(\sigma_i^x)^{a_{xi}}(\sigma_i^z)^{a_{zi}},
\end{equation}
where $a_{xi},a_{zi}\in\{0,1\}$. As an example, $IZXY=(0011|0101)$. Using the standard inner product on $2m$-dimensional vectors in $\mathbb{Z}_2$, resulting in a value 0 or 1, the following commutation relation holds
\begin{equation}
\label{commutation}
    \sigma(a_x|a_z)\sigma(b_x|b_z)=(-1)^{a_x\cdot b_z+a_z\cdot b_x}\sigma(b_x|b_z)\sigma(a_x|a_z).
\end{equation}
Following this, a scalar product $\times$ is defined on the vector space of binary vectors of length $2m$ as
\begin{equation}
\label{vectorproduct}
     a\times b:=a_x\cdot b_z+a_z\cdot b_x\in\{0,1\}.
\end{equation}
This vector product is symmetrical and linear in both $a$ and $b$. By evaluating this vector product it can easily be determined if the corresponding operators will commute or anti-commute, see Eq.~\eqref{commutation}.\\\\
The set of $k$ Pauli operators $\boldsymbol{\sigma}^{1},..., \boldsymbol{\sigma}^k$ appearing in the decomposition of $H$ can be represented by vectors $(a_x^1|a_z^1),...,(a_x^k|a_z^k)$ which can in turn be represented in a $k\times 2m$ matrix $G$ as
\begin{equation}
    G=\left[ \begin{array}{ccc|ccc}
       \cdots & a_{x}^1 & \cdots & \cdots & a_{z}^1 & \cdots  \\
       \cdots & a_{x}^2 & \cdots & \cdots & a_{z}^2 & \cdots  \\
        & \vdots &  & & \vdots &  \\
        \cdots & a_{x}^r & \cdots & \cdots & a_{z}^r & \cdots  \\
     \end{array}\right]
\end{equation}
Now if a Pauli operator $\sigma(b_x|b_z)$ were to commute with every term in $H$ it must be that $G\times(b_x|b_z)=0$, where the vector product $\times$, defined in Eq.~\eqref{vectorproduct}, is taken for every row in $G$. From the theory of stabilizer codes the concept of a parity check matrix $E$ is borrowed \cite{stabilizertheory}. This matrix $E$ is defined as 
\begin{equation*}
    E=\left[ \begin{array}{ccc|ccc}
       \cdots & a_{z}^1 & \cdots & \cdots & a_{x}^1 & \cdots  \\
       \cdots & a_{z}^2 & \cdots & \cdots & a_{x}^2 & \cdots  \\
        & \vdots &  & & \vdots &  \\
        \cdots & a_{z}^r & \cdots & \cdots & a_{x}^r & \cdots  \\
     \end{array}\right]
\end{equation*}
Note that $EG^{T}=0$. Now if the operator $\sigma(b_x|b_z)$ were to be part of the symmetry group $S$ then $(b_x|b_z)\in \text{ker}(E)$. Next a basis for $\text{ker}(E)$ can be constructed. Let $d:=\text{dim}(\text{ker}(E))$. These basis vectors $b^1,...,b^d$ give rise to Pauli operators $\boldsymbol{\sigma}^1,...,\boldsymbol{\sigma}^d\in P_m$. However, it is not guaranteed that these Pauli operators mutually commute and thus will form a set of generators for a symmetry group $S$. Commutation can be enforced by applying the algorithm below to obtain a set of new vectors $g^1,...,g^r\in \text{span}(b^1,...,b^d)$ with $r\leq d$. This is the part not formally justified in the paper by Bravyi et al \cite{taper}. The operators corresponding to these vectors will be linearly independent and commuting. We are thus looking for generators of a maximal Abelian subgroup of $\text{ker}(E)$. Hence
\begin{equation}
    \forall i,j\in\{1,...,r\}: g^i\times g^j=0.
\end{equation}
The algorithm for finding maximal abelian subgroups \cite{maxabelgroup} is described as follows
\begin{enumerate}
\item Define $G=\langle\sigma(b^1),...\sigma(b^d)\rangle$. 
\item Now compute the center $C$ of $G$ and set $A:=C$. Since $A$ is abelian $A\subseteq C_G(A)$. 
\item If $C_G(A)=A$ then $A$ is a maximal abelian subgroup of $G$. Else, take an arbitrary $a\in C_G(A)\backslash A$. Define $A:=\langle A\cup \{a\}\rangle$ and repeat step 3.
\end{enumerate}
This procedure must halt as $G$ is finitely generated and $A$ keeps growing while still $A\subseteq G$ holds. As mentioned previously, it is known that there are maximally $m$ generators in $S$ \cite{Nqbitsnec}. Thus $S=\langle \tau_1,\tau_2,...\tau_r\rangle$ where $r \leq m$. \\

At this point in the proof, Ref.~\cite{taper} claims that in all considered examples the symmetry generators $\tau_1,...,\tau_r$ are $Z$-type
Pauli operators and in addition that a subset of qubits $q(1), . . . , q(k)$ can be chosen such that
\begin{equation}
\sigma_{q(i)}^{x} \tau_{j}=(-1)^{\delta_{i, j}} \tau_{j} \sigma_{q(i)}^{x} .
\end{equation}
However, in the following, we find that this claim is nontrivial, although we will show that the assumption that the generators are $Z$-type is unnecessary.
 
\begin{definition}{Generator transform}\\
\label{def:gentransform}
Let $i\in{1,...,r}$ and let $\tau_i$ be a generator of a symmetry group $S=\langle\tau_1,...,\tau_r\rangle$. Let $\tau_i$ anti-commute with a single qubit operator $\sigma_j^\rho$, where $j\in{1,...m}$ and $\rho\in\{X,Y,Z\}$. Define new operators $\tilde{\tau}_k$ as
\begin{equation}
    \tilde{\tau_k} = \begin{cases}
               \tau_k \quad\quad\quad \mathrm{if} \quad k=i \vee \tau_k\sigma_j^\rho=\sigma_j^\rho\tau_k \\
               \tau_i\tau_k \quad\quad \mathrm{else}
            \end{cases}.
\end{equation}
These new operators are linearly independent, mutually commuting and generate $S$ according to lemma 2. Furthermore,  these new generators were transformed in such a way that $\tilde{\tau}_i$ anti-commutes with $\sigma_j^\rho$ while all other $\tilde{\tau}_k$ commute with $\sigma_j^\rho$ according to lemma 3. 
\end{definition} 

\noindent For each $n\leq r$ the qubit reduction scheme requires a triple $\{A_n, B_n, V_n\}$ with $A_n=\{q(i)|i=1,...,n\quad i\not=j \rightarrow q(i)\not=q(j) \}$ a set of qubit-indices, $B_n=\{\rho(i)| i=1,...,n \quad \rho(i)\in\{X,Z\}\}$  a set of labels and $V_n=\{\nu_i|i=1,...,r, \quad \nu_i=\tau_1^{\alpha_1}...\tau_r^{\alpha_r}, \alpha_j\in \{0,1\} \}$ a set of generators of $S$. This triple $\{A_n, B_n, V_n\}$ should satisfy the predicate
\begin{equation}
\label{eq:predicament}
\begin{aligned}
    \Pi(\{A_n,B_n,V_n\}): &\quad \forall i\in \{1,..., n\} \hspace{0.5em}\forall j\in\{1,..., r\} \\
    &\sigma^{\rho(i)}_{q(i)}\nu_j=(-1)^{\delta_{i,j}}\nu_j\sigma^{\rho(i)}_{q(i)}
    \end{aligned}
\end{equation}
This requirement states that each qubit $i$ should anti-commute only with generator $i$.

\begin{lemma}
\label{lagrangelemma}
Assume that $\Pi(\{A_n,B_n,V_n\})$ is satisfied. Let $\tau_{n+1}=P_1P_2...P_m$ where $P_i\in\{I,X,Y,Z\}$. There exists $j\in\{1,...,m\}$ such that $j\not\in A_n$ and $P_j\not= I$.
\end{lemma}
\begin{proof}
Proof by contradiction. Assume that $\forall i\not\in A_n: P_i=I$. Define the subgroup $S_l^n<P_m$ for $l\leq n$ as
\begin{equation}
\begin{aligned}
    &S_l^n\left(\{A_n,B_n,V_n\}\right)=\\
    &\{\sigma=P_1P_2...P_m |\hspace{0.5em} j\not\in A_n \Rightarrow P_{j}=I,\hspace{0.5em} \forall i\in\{1,..., l\} \\
    & \sigma\tau_i=\tau_i\sigma, \hspace{0.5em}  \forall w\in\{1,..., n\} \hspace{0.5em}  \sigma\sigma^{\rho(w)}_{q(w)}=\sigma^{\rho(w)}_{q(w)}\sigma\}
    \end{aligned}
\end{equation}
$S_l^n$ is a subgroup of $P_m$ as for every $l\leq n$ it is true that $I^{\otimes m}\in S_l$ and if $\sigma_a,\sigma_b\in S_l^n$ then $\sigma_a\sigma_b\in S_l^n$ (the inverses are of course trivially satisfied as every element is its own inverse). One can see that $\tau_{n+1}\in S^n_n$ by construction. From the condition $\forall w\in\{1,...,n\}\quad \sigma\sigma^{p_{q(w)}}_{q(w)}=\sigma^{p_{q(w)}}_{q(w)}\sigma$ it follows that $|S^n_0|=2^n$, as in every position in $A_n$ only two operators out of $\{I,X,Y,Z\}$ are allowed. It is easy to see that $S^n_{l+1}$ is a subgroup of $S^n_{l}$. Lagrange's theorem gives
\begin{equation}
|S^n_{l}|/|S^n_{l+1}|=|S^n_{l}/S^n_{l+1}|\in \mathbb{N}.    
\end{equation}

\noindent Note now that $\sigma^{\rho(l+1)}_{q(l+1)}$ is in $S^n_{l}$ but not in $S^n_{l+1}$, since it must anti commute with $\tau_{l+1}$. So, $|S^n_{l}|/|S^n_{l+1}|\geq 2$. Then $|S^n_{0}|/|S^n_{n}|\geq 2^n$. Therefore, $|S_{n}|=1$ and thus $S_{n}={I}$, the only group with one element. But by construction it was known that $\tau_{n+1}\in S^n_n$. Since $I$ can not be a generator a contradiction is reached. Therefore, there must indeed exist a $j\in{1,..., m}$ such that $j\not\in A_n$ and $P_{j}\not= I$. 
\end{proof}

\begin{theorem}
\label{the:5}
Let $S=\langle \tau_1,\tau_2,...\tau_r\rangle$ be a symmetry group. For every $n\leq r$ there exists a triple $\{A_n,B_n,V_n\}$ such that $\Pi(\{A_n,B_n,V_n\})$ is satisfied.
\end{theorem}

\begin{proof}
The proof is given by induction over $n$
\\\\
\noindent \textbf{Base case $n=1$:} consider $\tau_1=P_1P_2...P_m$ where $P_i\in\{I,X,Y,Z\}$. Since $\tau_1$ is a generator there must be a $P_j\not=I$. The first qubit index is chosen as $q(1)=j$ thus $A_1=\{j\}$. The first label $\rho(1)$ is chosen as
\begin{equation}
    \rho(1)= \begin{cases}
               X \quad\quad \mathrm{if} \quad P_j\in\{Y,Z\}\\
               Z \quad\quad \mathrm{if} \quad P_j=X
            \end{cases}.
\end{equation}

\noindent Thus, $B_1=\{\rho(1)\}.$ The generator transform from definition~\ref{def:gentransform} can now be applied to the original generators of $S$ to yield new generators $\tilde{\tau}_1,...,\tilde{\tau}_r$ which form $V_1$. For ease of argumentation the new generators $\tilde{\tau}_1,...,\tilde{\tau}_r$ are renamed to $\tau_1,...,\tau_r$. Lemmas 2 and 3 now show that $\Pi(\{A_1,B_1,V_1\})$ is satisfied.\\

\noindent \textbf{Induction cases $2\leq n+1\leq r$}. 

\noindent From lemma~\ref{lagrangelemma} there exists a $j\in{1,...m}$ such that $j\not\in A_n$ and $P_j\not= I$. One can then choose $q(n+1)=j$ and 
\begin{equation}
    \rho(n+1)= \begin{cases}
               X \quad\quad \mathrm{if} \quad P_j\in\{Y,Z\}\\
               Z \quad\quad \mathrm{if} \quad P_j=X
            \end{cases}.
\end{equation}

\noindent The generator transform of definition~\ref{def:gentransform} can be applied to the generators in $V_n$ to yield new generators $\tilde{\tau}_1,...,\tilde{\tau}_r$. These generators form $V_{n+1}$. For ease of argumentation the new generators $\tilde{\tau}_1,...,\tilde{\tau}_r$ are renamed to $\tau_1,...,\tau_r$. $\Pi(\{A_{n+1},B_{n+1},V_{n+1}\})$ is now satisfied.
\end{proof}

\noindent The existence and construction of a triple $\{A_r,B_r,V_r\}$ such that $\Pi(\{A_r,B_r,V_r\})$ is satisfied has been proven.

\noindent Using the found triple $\{A_r,B_r,V_r\}$ one can define the operators $U_i$ for $i=1,...r,$ according to
\begin{equation}
    U_i=\begin{cases}
               \frac{1}{2}(\sigma^{z}_{q(i)}+\tau_i)(\sigma^{z}_{q(i)}+\sigma^x_{q(i)}) \quad \mathrm{if} \quad \rho(i)=Z\\
               \frac{1}{\sqrt{2}}(\sigma^{x}_{q(i)}+\tau_i) \quad\quad\quad\quad\quad\quad\hspace{0.2 em} \mathrm{if} \quad \rho(i)=X
            \end{cases}.
\end{equation}
Below certain properties of the defined operators $U_i$ are shown. Since Pauli matrices are Hermitian and unitary $\sigma^{\rho(i)\dagger}_{q(i)}=\sigma^{\rho(i)}_{q(i)}$, $\tau_i^\dagger=\tau_i$ and $\sigma^{\rho(i)\dagger}_{q(i)}\sigma^{\rho(i)}_{q(i)}=\tau_i^\dagger\tau_i=1$.  Using the implied commutation relations  the following holds 

\begin{equation}
\begin{aligned}
    &U_i^\dagger U_i=1, \quad U_i\sigma^x_{q(i)}U_i^\dagger=\tau_i,\\
    &U_i\sigma^x_{q(j)}U_i^\dagger=\sigma^x_{q(j)} \quad \mathrm{for} \quad j\not=i.
    \label{propertiesu}
\end{aligned}
\end{equation}

This allows for the definition of the unitary operator $U=U_1U_2...U_rW$, where $W$ is a permutation operator assigning qubit $i$ to qubit $q(i)$. As a combination of unitary operators $U$ is also unitary. The following relation follows from Eq.~\eqref{propertiesu} $U\sigma^x_{i}U^\dagger=\tau_i.$

The transformed Hamiltonian $H'$ can be defined as
\begin{equation}
    H'= U^\dagger HU=\sum_k h_k U^\dagger\boldsymbol{\sigma}^kU.
\end{equation}
Define $\boldsymbol{\eta}^k=U^\dagger\boldsymbol{\sigma}^kU$. Since $\forall i,j\in\{1,...,r\}$ $[\boldsymbol{\sigma^i},\tau_j]=0$ it must be that $\forall i,j\in\{1,...,r\}$ $ [\boldsymbol{\eta^i},\sigma^x_j]=0$. Therefore, $H'$ commutes with the last $r$ qubits and thus shares eigenvectors with $\sigma^x_{m-(r-1)},...\sigma^x_{m}$. The last $r$ qubits can thus be replaced by the $X$ operator eigenvalues $\pm1$.\\

Since the proof presented in this section is constructive, it should be feasible to implement this scheme in practice. For instance, in Ref.~\cite{taper} this scheme has been applied successfully to Hamiltonians of small molecules such as BeH$_2$ and HCl to reduce the number of qubits necessary to describe the problem by four qubits, which is two more than just those from electron and spin number. In the remainder of the paper, we focus on the importance of the role that entangling methods play in the convergence of the VQE. We analyze only relatively simple systems that posses no additional $\mathbb{Z}_2$ symmetries other than the electronic and spin symmetries present.
\newpage

\section{Multi-qubit entangling methods}
\label{sec:section5}
In the preparation of the trial state the entanglement operator $U_{ent}$ ensures that the full Hilbert space of trial states can be reached. In this section, a representative group of multi-qubit entangling methods will be described that will be used in the VQE simulations of Sec. VIII.  Some of the entanglers are gate based while others are native to system interactions \cite{a_otherVQE}. With regards to the latter, we study in particular the Rydberg quantum computing system and its accompanying Rydberg interaction \cite{rydberg1,rydberg2,Jaksch2000,morgado2020quantum}. The goal here is to show that entanglement gates based on these interactions can reach similar energy differences with the ground state as standard multi qubit entanglement gates in the form of combinations of CNOT gates. The advantage of these interaction based entanglement operations is that they are native to a system and thus do not require extensive constructions or highly specific schemes. To illustrate the comparable convergence between these methods we look at a small ensemble of interaction based/native entanglement operations and compare this to standard combinations of CNOT gates.

\begin{figure}[H]
\includegraphics[scale=0.50]{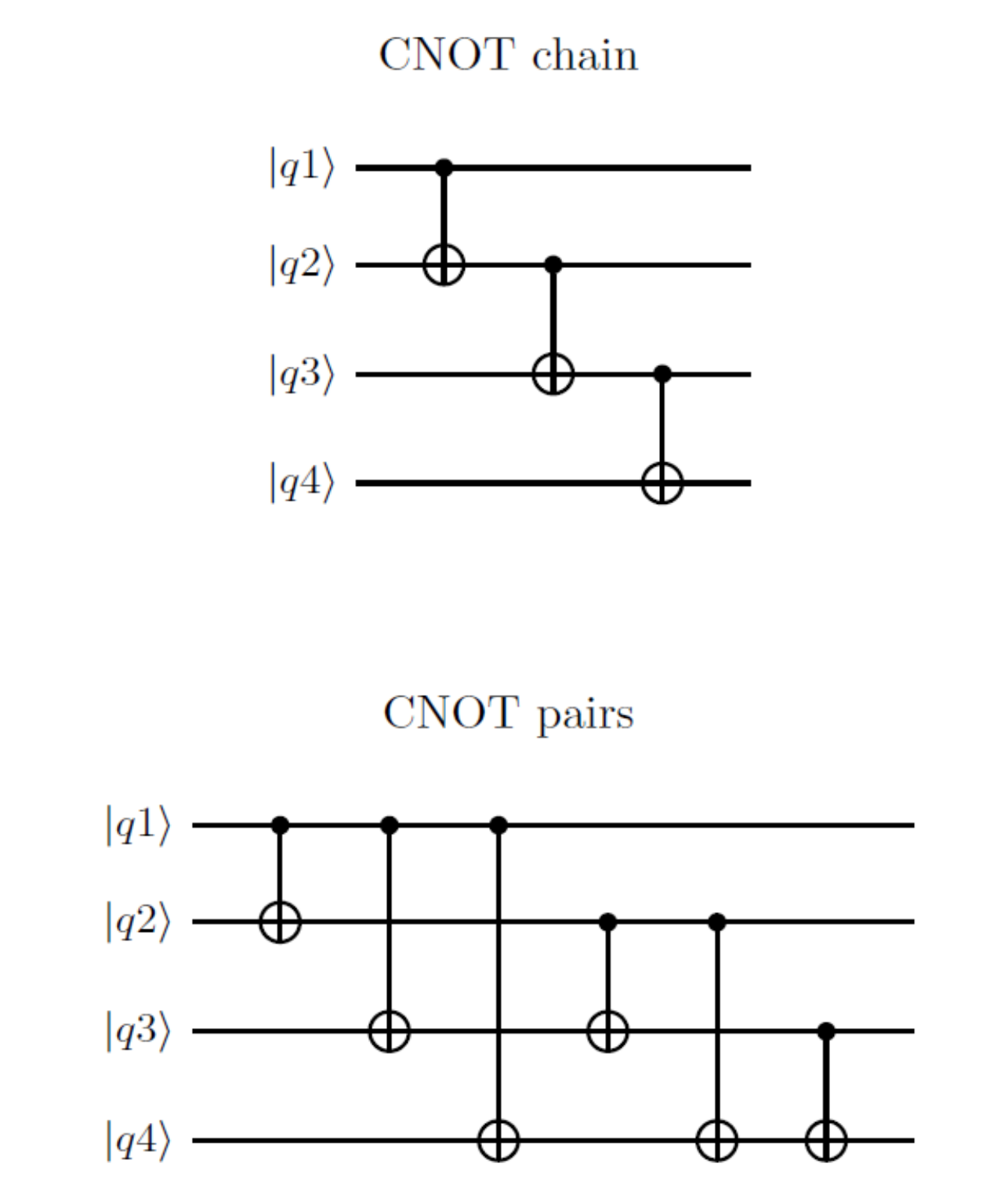}
\caption{Quantum circuit diagram of the CNOT chain and pairs methods for $m=4$.}
\label{fig:cnotchain}
\end{figure}

\newpage
The standard entanglement operator on two qubits is the CNOT gate. A way to entangle more than two qubits is by applying the CNOT gate in a chain or pairwise as is shown in Fig.~\ref{fig:cnotchain}.

Since qubit operations are prone to errors it is desirable to use as few as possible to entangle all qubits. The chain and pair methods respectively use $m-1$ and $m(m-1)/2$ CNOT operations for $m$ qubits. A way of reducing the number of operations is using entanglement interactions native to the qubit system. These interactions can entangle all qubits using only one operation and thus when considered as a single gate, are a more efficient way of entangling all qubits. Note that this efficiency only holds if the interaction is native to the qubit system, on other systems it might have to be constructed out of multiple other interactions.\\ 

A method proposed by Jaksch et al. \cite{Jaksch2000} uses the Rydberg interaction. The Rydberg interaction is a long range interaction which, when one Rydberg atom (qubit) in the excited state, blocks others from that state. Thus, it can function much like a CNOT gate. However, because of the long range interaction it also allows for a C$_{m-1}$NOT gate where one qubit controls all $m-1$ other qubits, thus multi-qubit entanglement is achieved \cite{HadamardRydberg, rydbergqubit}. Another physical multi-qubit entangling method is the Krawtchouk chain  \cite{Groenland}. The Krawtchouk chain is an entangling method based on the physical Hamiltonian of a Rydberg system which can be described as an interacting spin chain.  Such a spin chain could consist of trapped ions or ultra-cold Rydberg atoms \cite{Krawtchouk}. The Krawtchouk chain allows for the construction of the PST$_m$ gate which mirrors the left and right side of a chain of qubits and multiplies it by a factor $i$ (for instance $|01101\rangle\rightarrow i|10110\rangle$). This effect is often called perfect state transfer (PST).\\

The Krawtchouk chain is an important method to construct other multi-qubit gates.  Several studies have shown that the Krawtchouk Hamiltonian gate can be transformed, using only few one and two qubit rotations \cite{Groenland,Krawtchouk2}.  Two examples of these transformations are the iSWAP$_2$ gate and the C$_{m-1}$NOT gate. The iSWAP$_2$ gate swaps the last two qubits and multiplies the state by a phase $i$.  This swapping is controlled by the first $m-2$ qubits.\\

In the quantum chemistry problems considered, some combinations of orbitals interact more than others. Similarly, in some of these entanglement schemes, such as CNOT pairs and C$_{m-1}$NOT, the qubits play different roles, e.g. control versus target. This means that different assignments of orbitals to qubits may result in different convergence. Based on prior information on the problem from quantum chemistry, it might thus be beneficial to assign certain qubits to certain orbitals. Since this paper does not focus on initial chemistry problem analysis, we simply assign orbital 1 to qubit 1, etc.
\newpage

\section{Results}
\label{sec:section7}
In this section, we investigate the behaviour of the discussed multi-qubit entangling methods and classical optimization algorithms on VQE problems of H$_2$, LiH and H$_2$O. The simultaneous perturbation stochastic approximation (SPSA) algorithm, which is described in App.~\ref{app:appendixSPSA}, is used as the classical optimization algorithm throughout. For problems with a small number of qubits, such as considered here, the dimension of the state space is limited. In other more complex problems it may be advantageous instead to start off with a specific initial state to check if the VQE algorithm is able to improve upon this. For consistency, we choose the initial state to be the vacuum state $|\psi_{init}\rangle=|0...0\rangle$ in this work. When the found energies $E_{VQE}$ of two methods are compared to the ground state energy $E_g$, the method with the lower difference $E_{VQE}-E_g$ is said to outperform the other.

\begin{figure}[H]
\centering
  \includegraphics[width=\linewidth]{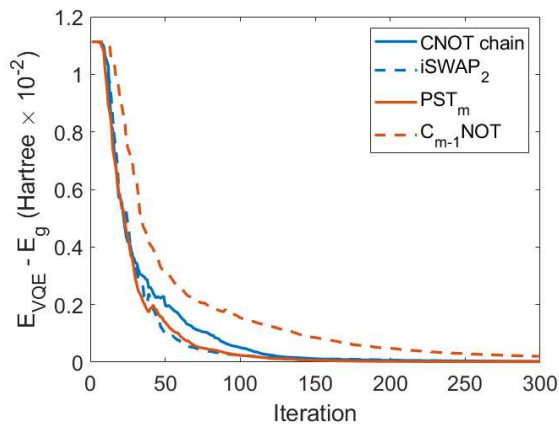}
\caption{A typical energy difference from the ground state energy versus iteration number for the entangling methods PST$_m$, C$_{m-1}$NOT and iSWAP$_2$ compared to the standard CNOT chain. The VQE problem describes a LiH molecule (6 qubits) at depth 6, with classical optimization SPSA.}
\label{fig:krawtchoukchainbased2}
\end{figure}

\begin{figure*}
\captionsetup{justification=raggedright}
    \begin{minipage}[t]{.49\textwidth}
         \includegraphics[width=0.89\columnwidth]{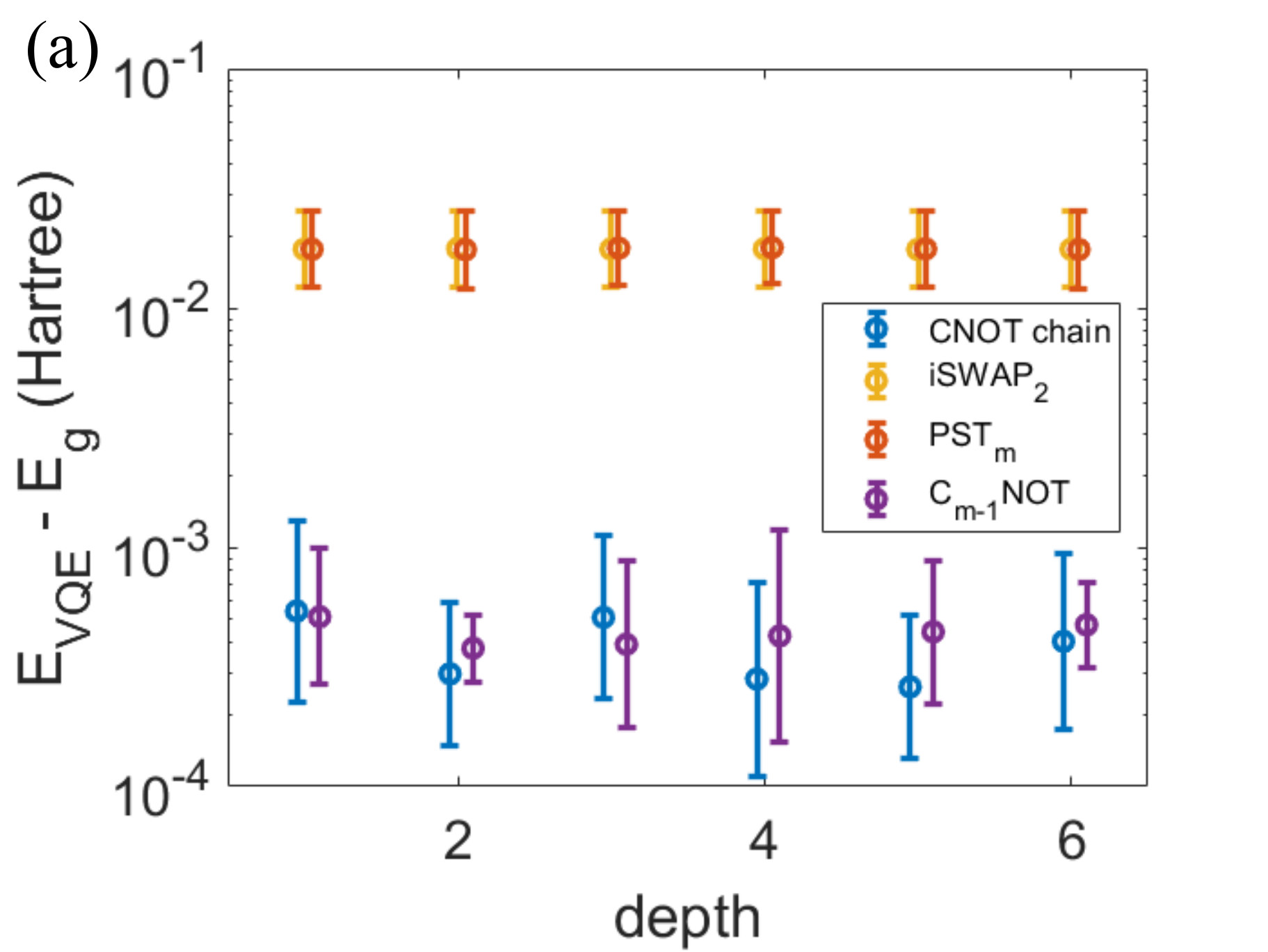}\\
    \end{minipage}%
    \begin{minipage}[t]{.49\textwidth}
         \includegraphics[width=0.89\columnwidth]{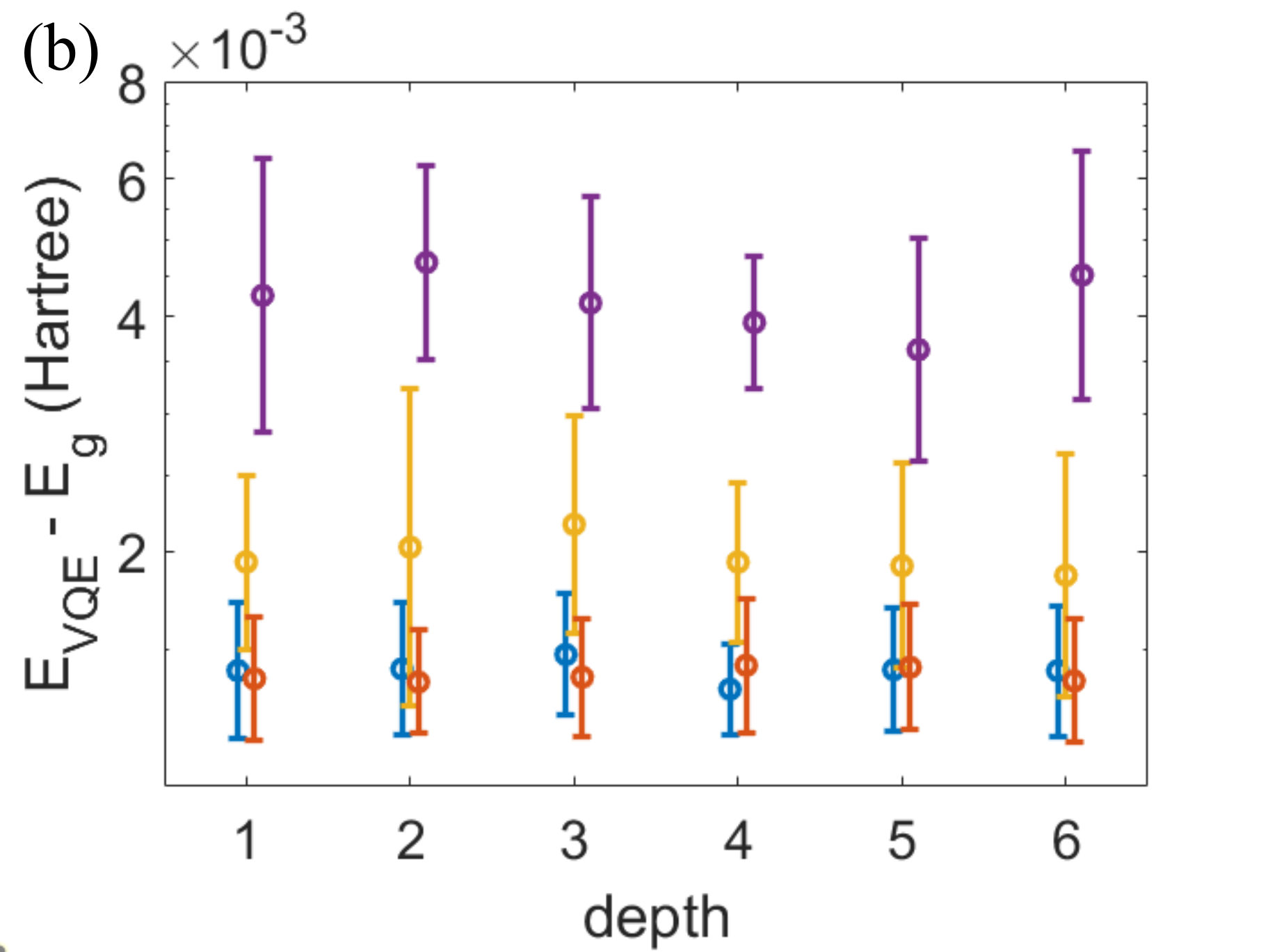}\\
    \end{minipage}%
\caption{Differences from the exact ground energy vs. depth for $10^3$ iteration SPSA simulations comparing the entangling methods for H$_2$ (a) and LiH (b) at depths 1 to 6. Energy differences are calculated  using the geometric average over several internal configurations of the molecule. The spread of each data point is given by the geometric standard deviation. Depth values are slightly offset for readability.}
    \label{fig:depthh2lih}
\end{figure*}

Specifically, H$_2$ is modelled with a 1s orbital for each atom resulting in a 4 qubit system which is reduced to a 2 qubit system by the $\mathbb{Z}_2$ symmetries. We run the simulations over  interatomic distances from 0.392$a_0$ to $0.931a_0$ in steps of 0.049$a_0$, with $a_0$ the Bohr-radius. The LiH molecule is aligned on the $x$-axis and we consider the $1s$, $2s$ and $2p_x$ orbitals on the Li atom and a $1s$ on the H atom resulting in 8 qubits which is again reduced to 6. Simulations are run over interatomic distances from 0.72$a_0$ to $1.71a_0$, in steps of 0.09$a_0$. For both molecules the resulting errors with the ground state energy are geometrically averaged over the chosen bond length.\\

To investigate the dependence of the VQE results on entangling methods, we first analyze the ground state of LiH for d=6 using the PST$_m$, iSWAP$_2$, C$_{m-1}$NOT and CNOT chain gates (discussed in Sec.~\ref{sec:section5}) as shown in Fig.~\ref{fig:krawtchoukchainbased2}.
Here it is found that methods based on the Rydberg interaction are able to perform comparably to the CNOT chain method (see Sec.~\ref{sec:section5}.).  Repeating these simulations for LiH and H$_2$ as a function of depth $d$ we get the resulting energy differences given in Fig.~\ref{fig:depthh2lih}. These results seem to further support that Rydberg interaction entangling methods can compare to the CNOT chain method in terms of performance. The optimal entangling method however appears to be problem dependent. The entanglement depth $d$ can open up the search space for the trial state allowing lower energy differences to be reached as can be seen in many of the LiH cases of Fig.~\ref{fig:depthh2lih}. However, it can also overcomplicate the problem by introducing too many rotation parameters and thus take more iterations to converge as we see in H$_2$ cases. Another interesting fact is seen for PST$_m$ and iSWAP$_2$ at H$_2$, where the entangling method is not fit for the problem regardless of the depth, and an energy difference as low as the CNOT chain and C$_k$NOT differences is not reached. \\ 

To see that the preference of one entangling method over the other is not the result of the geometric averaging of the energy difference we solve the VQE problem for H$_2$O (8 qubits) at bond length $0.99a_0$ with angles between the H atoms ranging from $\varphi=0.2\pi$ to $\varphi=0.7\pi$ in steps of $0.1\pi$ (the angle with the lowest ground state energy for H$_2$O is $\varphi\approx0.58\pi$). Doing so allows us to determine the influence of a slightly changed internal configuration of the molecule on the VQE convergence and entangler performance. The H$_2$O molecule lies in the $xy$-plane, the 1s orbital is filled for the O atom, the $2s$, $2p_x$ and $2p_y$ are active, and a 1s orbital is considered active on each of the H atoms. After reduction, this results in 8 qubits. Fig.~\ref{fig:h2oresults} shows the results per angle while Fig.~\ref{fig:h2oresultsaverage} shows the energy differences geometrically averaged over all angles together with the geometric standard deviation. The results in Fig.~\ref{fig:h2oresults} further support the claim that Rydberg based entanglers can perform as well as the CNOT chain. Furthermore, in Fig.~\ref{fig:h2oresults} it is seen that for all angles the iSWAP$_2$ and PST$_m$ gates tend to converge faster than the CNOT chain and C$_{m-1}$NOT gates. We also see that the C$_{m-1}$NOT method tends to converge the slowest. At the $\varphi=0.5\pi$ angle we see the CNOT chain getting stuck in a local minimum and not able to decrease the energy any further. The fact that iSWAP$_2$ and PST$_m$ perform consistently well at slightly different internal configurations shows that some entanglement methods are more suited to a certain molecular problem than others. 

\begin{figure*}
\captionsetup{justification=raggedright}
    \begin{minipage}[t]{.32\textwidth}
         \includegraphics[width=\columnwidth]{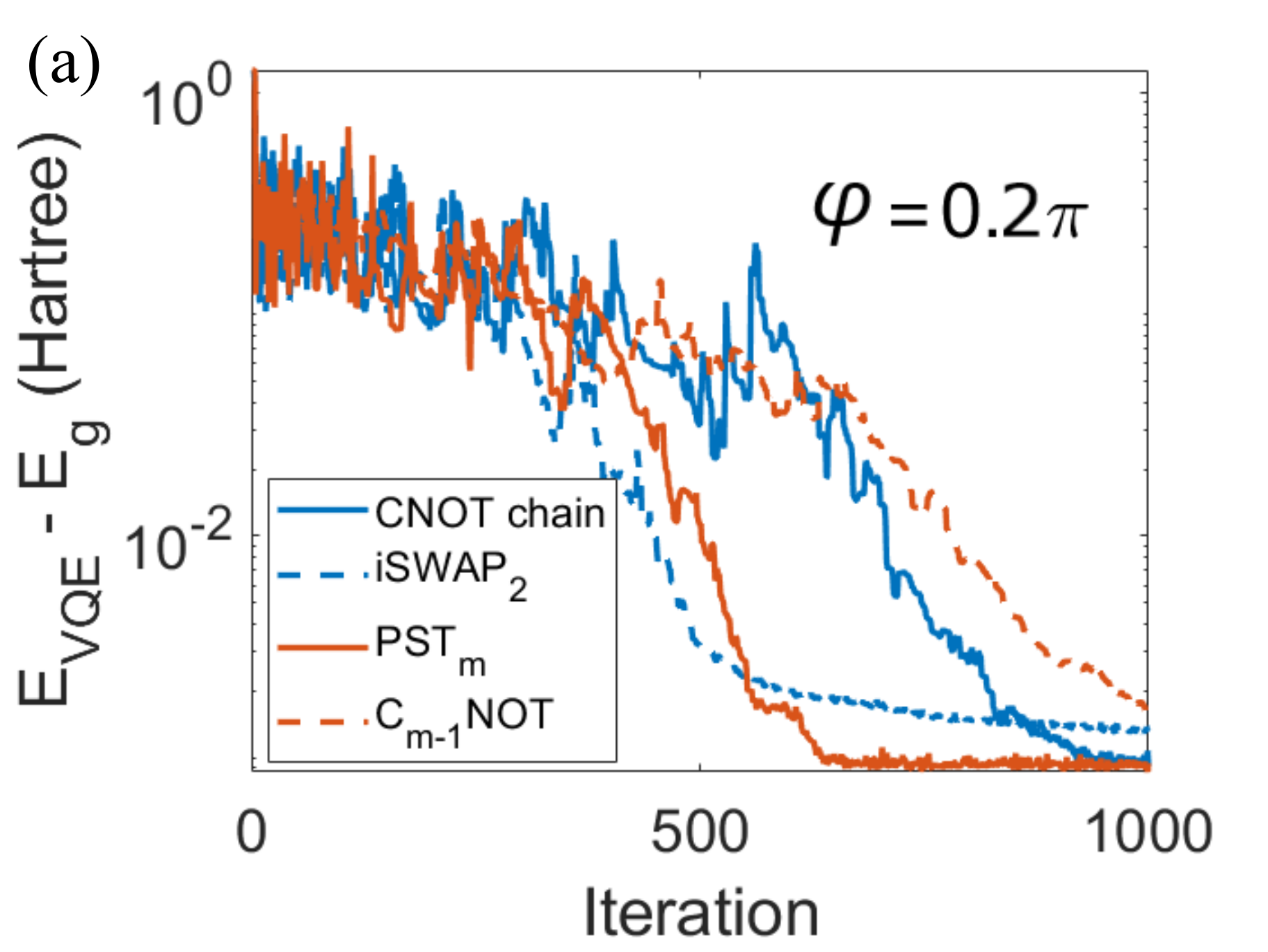}
    \end{minipage}%
    \begin{minipage}[t]{.32\textwidth}
         \includegraphics[width=\columnwidth]{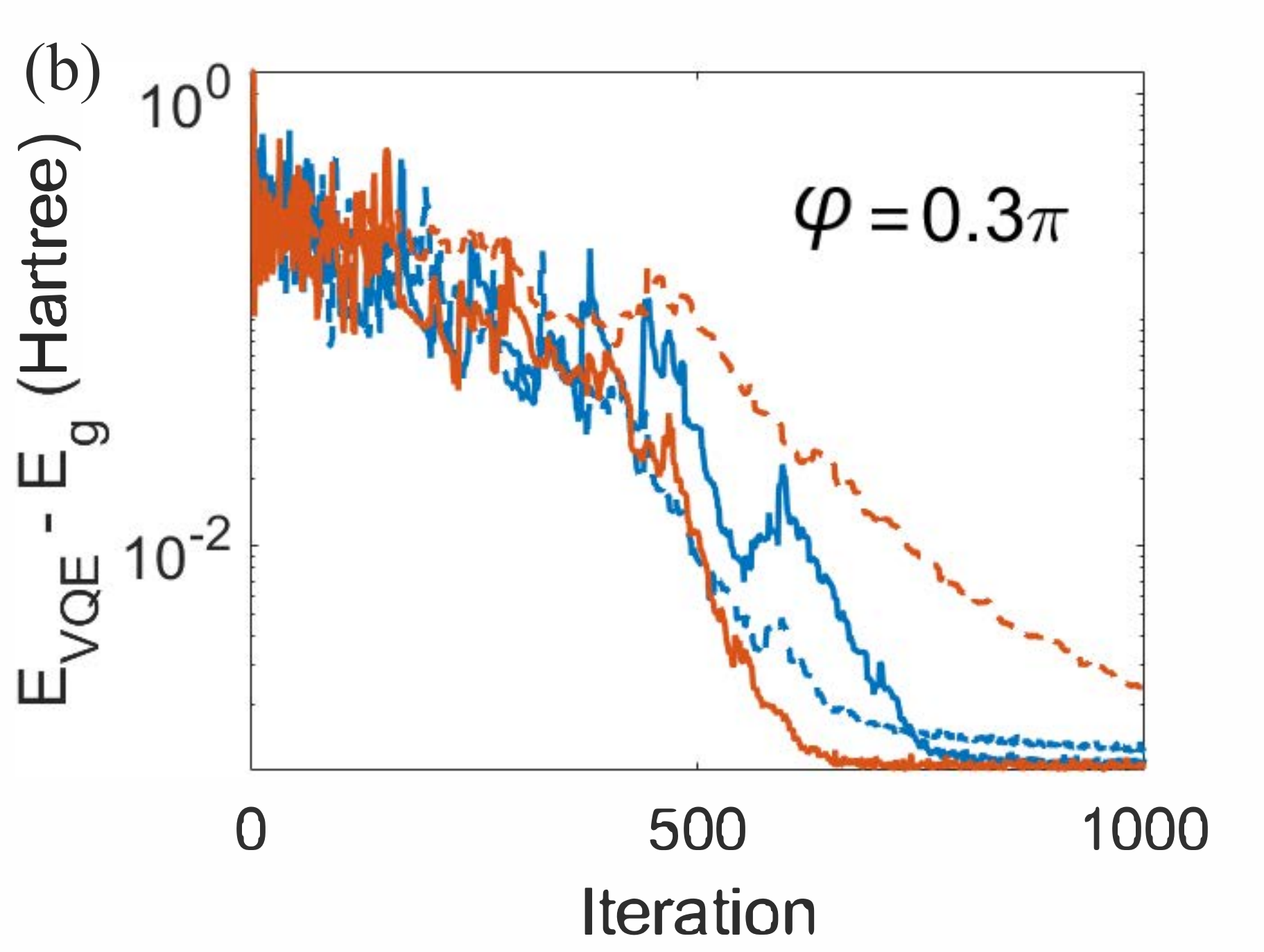}
    \end{minipage}%
    \begin{minipage}[t]{.32\textwidth}
         \includegraphics[width=\columnwidth]{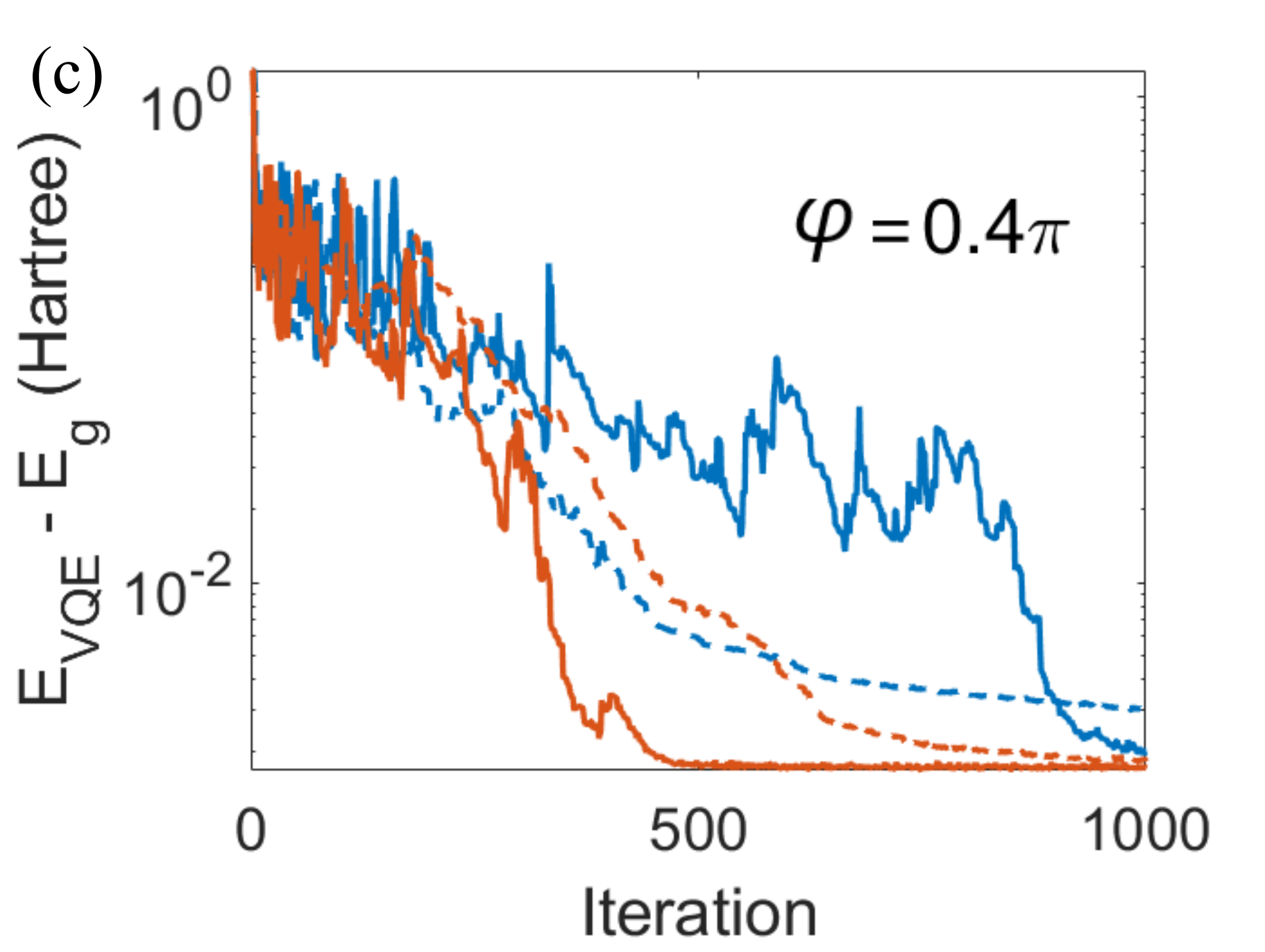}
    \end{minipage}%
    \\
    \begin{minipage}[t]{.32\textwidth}
         \includegraphics[width=\columnwidth]{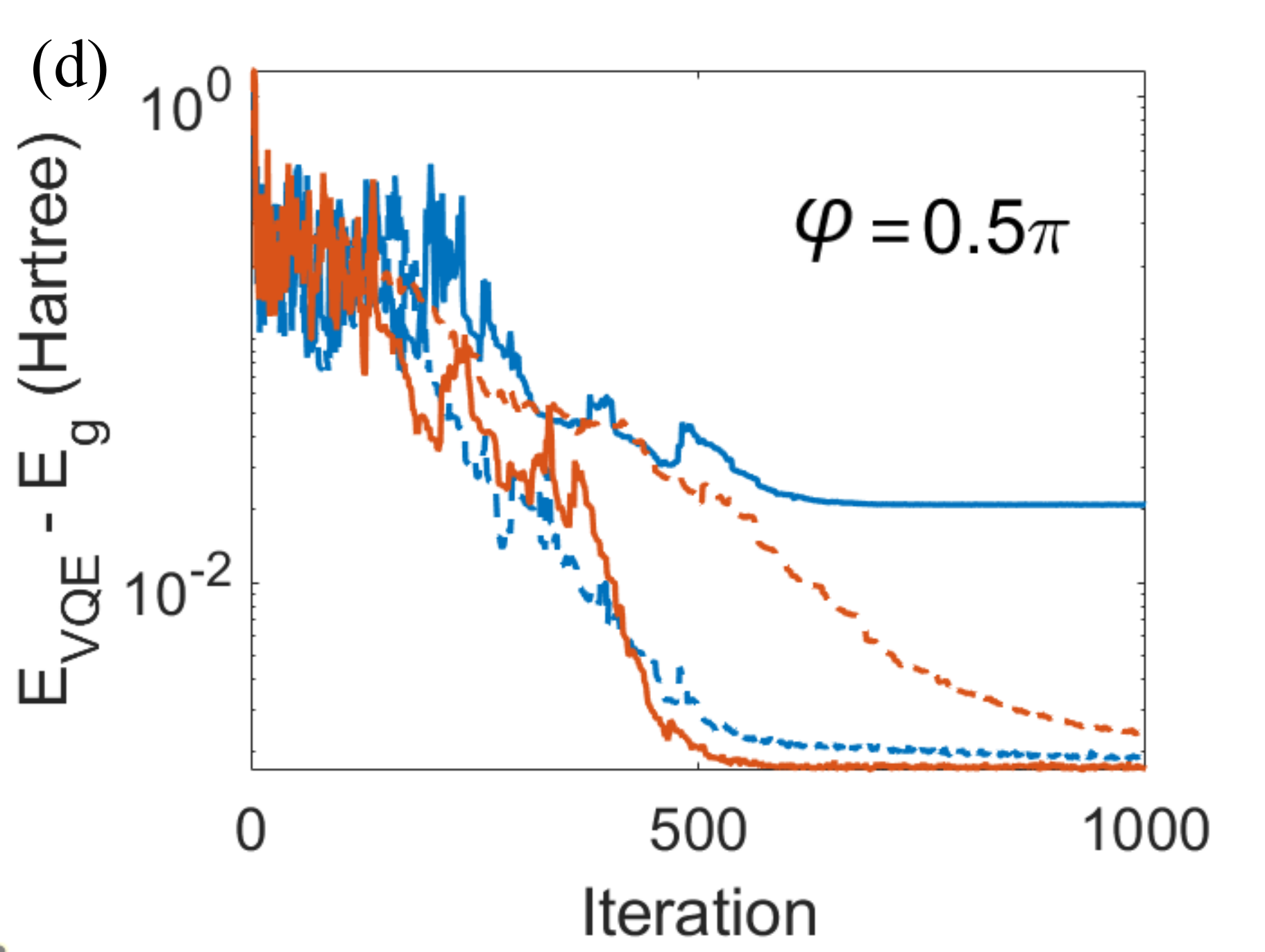}
    \end{minipage}%
    \begin{minipage}[t]{.32\textwidth}
         \includegraphics[width=\columnwidth]{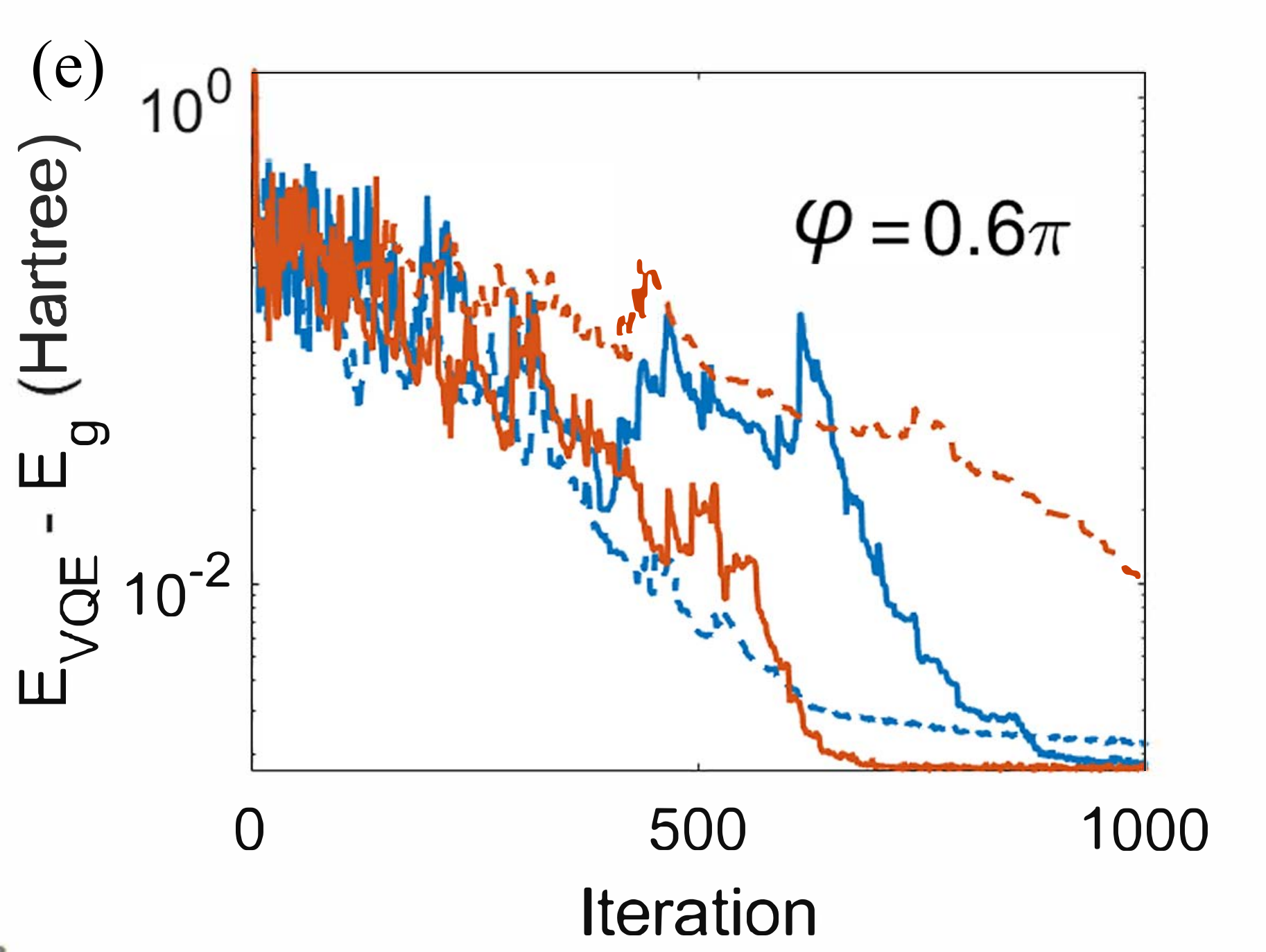}
    \end{minipage}%
    \begin{minipage}[t]{.32\textwidth}
         \includegraphics[width=\columnwidth]{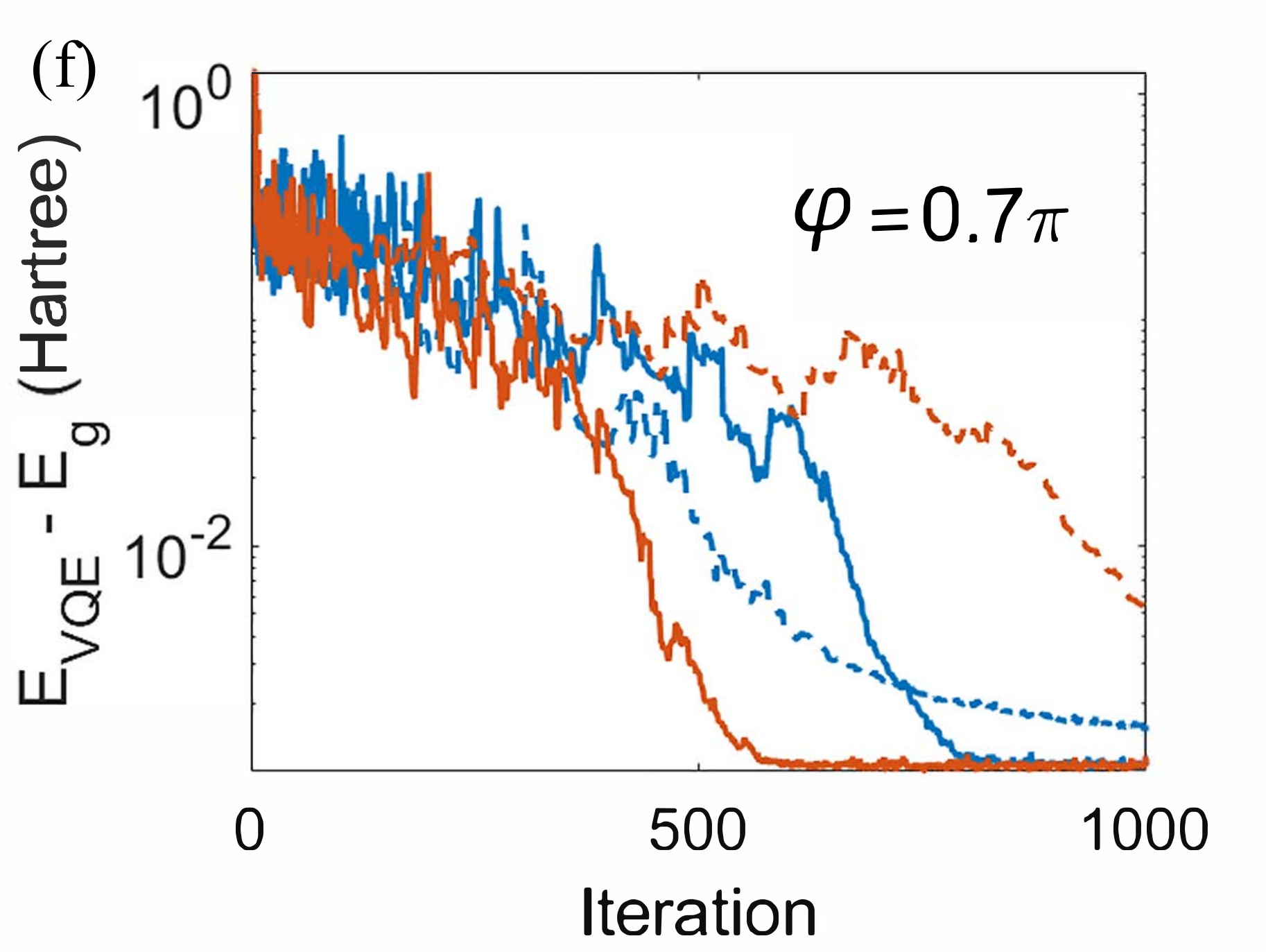}
    \end{minipage}
\caption{Iteration number versus energy difference with the 
    ground state energy for H$_2$O with depth $d=6$ at different angles $\varphi$ between the H atoms, optimized with the SPSA algorithm: (a) $\varphi=0.2 \pi$, (b) $\varphi=0.3 \pi$, (c) $\varphi=0.4 \pi$, (d) $\varphi=0.5 \pi$, (e) $\varphi=0.6 \pi$, and (f) $\varphi=0.7 \pi$.}
    \label{fig:h2oresults}
\end{figure*}

\begin{figure*}
\captionsetup{justification=raggedright}
    \begin{minipage}[t]{.49\textwidth}
         \includegraphics[width=\columnwidth]{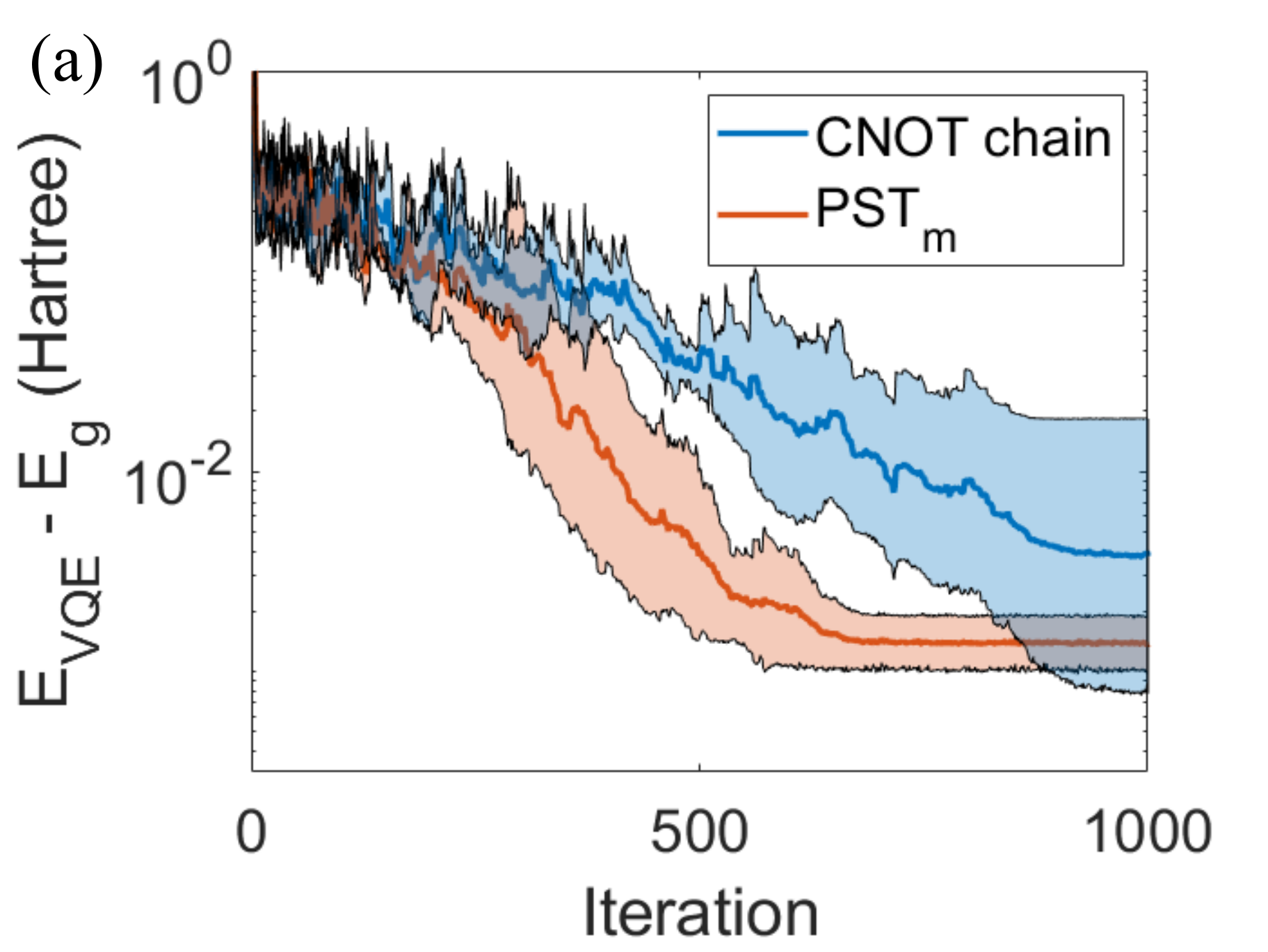}
    \end{minipage}%
    \begin{minipage}[t]{.49\textwidth}
         \includegraphics[width=\columnwidth]{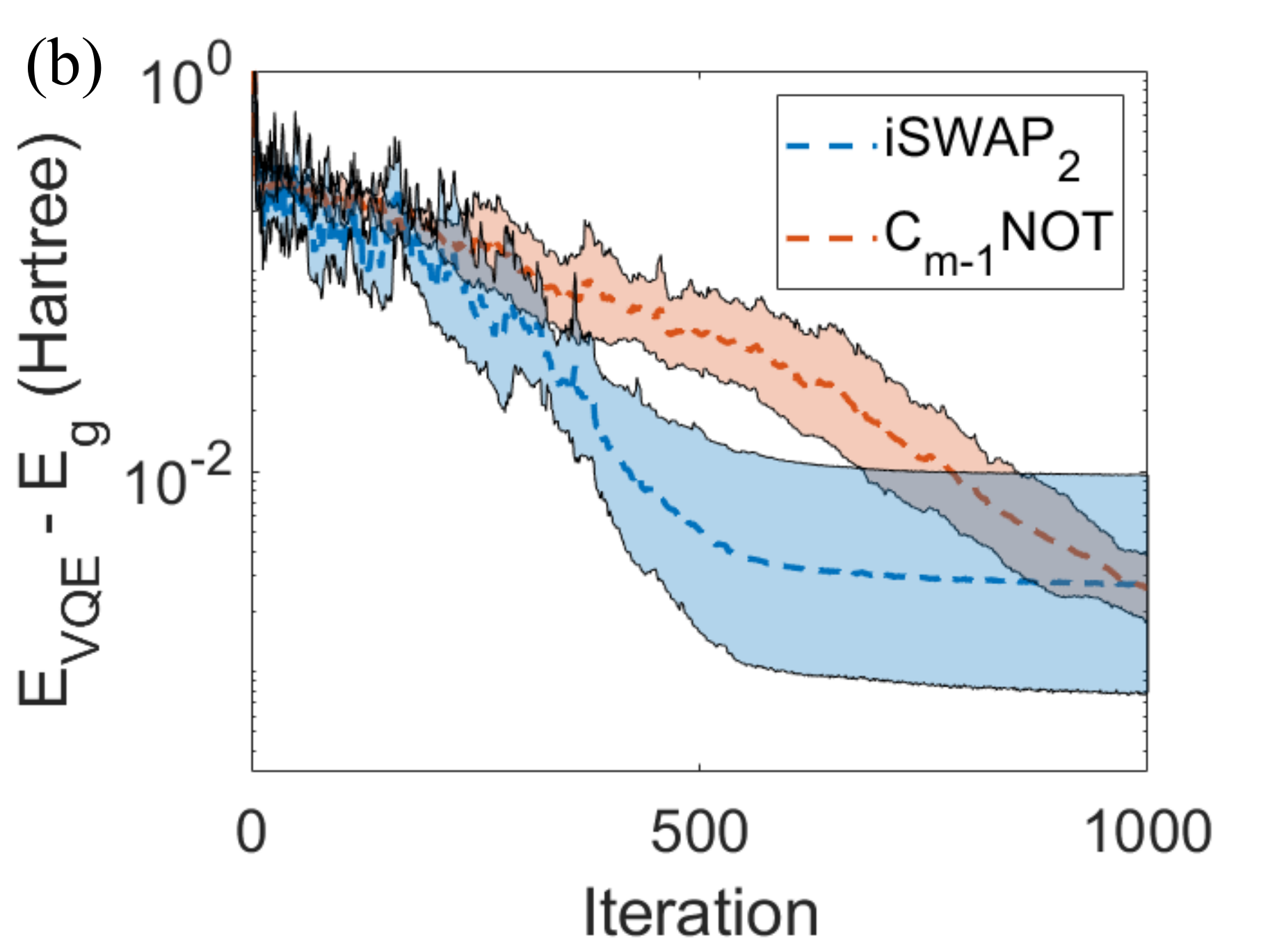}
    \end{minipage}%
\caption{Iteration number versus energy difference with the 
    ground state energy for H$_2$O with depth $d=6$ geometrically averaged over the different angles $\varphi$ between the H atoms, optimized with the SPSA algorithm. In (a) results are shown for the CNOT chain and PST$_m$ entanglement methods, in (b) results are shown for the ISWAP$_2$ and C$_{m-1}$NOT entanglement methods. Spread obtained via geometric variance.}
    \label{fig:h2oresultsaverage}
\end{figure*}

\newpage
\section{Conclusion}
\label{sec:section8}
In this paper, we have focused on optimizing the VQE algorithm for quantum chemistry applications. First,  a constructive proof was shown for a qubit reduction scheme using the $\mathbb{Z}_2$ symmetries of Hamiltonians, based on the paper by Bravyi et al. \cite{taper}, which effectively reduces the total number of necessary qubits. The multiple $\mathbb{Z}_2$ qubit reduction can thus be an important ingredient for reducing the total number of qubit manipulations for VQE problems.

After this the performance of the VQE algorithm as a function of the entangling method choice and depth was examined. Entangling methods based on both physical interactions, such as PST$_m$ and C$_{m-1}$NOT gave a convergence that was comparable to that of the CNOT chain method in the tested small molecule VQE problems. Looking at individual problems there seems to be preference for certain entangling operations over others: in H$_2$ the CNOT chain and C$_k$NOT performed well while in LiH and H$_2$O the C$_k$NOT and PST$_m$ gates were preferred. We thus see that even in VQE problems with low number of qubits the choice of entanglement operation plays an important role in the speed of convergence.
\\

To further investigate the link between entangling method and VQE problems, VQE calculations should be performed for more complicated molecular systems to see if a certain preference is maintained at a higher number of qubits. In this regime the barren plateau problem \cite{barrenplateau1} will have to accounted for. Another interesting point would be the prior analysis of the quantum chemistry problem to generate more optimal assignments from orbitals to qubits. Furthermore, our findings motivate future studies on the link between a molecular VQE problem and the entangling method. A measure of which entangling operation would do well prior to calculation would be valuable knowledge concerning the demands on hardware. Such a concept would be foundational towards the design of efficient VQE circuitry.

\onecolumngrid

\section*{Acknowledgements}
We thank Jasper Postema and Gijs Groeneveld for discussions. This research is financially supported by the Dutch Ministry of Economic Affairs and Climate Policy (EZK), as part of the Quantum Delta NL programme, and by the Netherlands Organisation for Scientific Research (NWO) under Grants No. 680-47-623 and 680-92-18-05.

\section*{Conflict of interest}
The authors have no conflicts to disclose.

\section*{Data Availability}
The data that support the findings of this study are available from the corresponding author upon reasonable request.

\newpage
\twocolumngrid

\newpage
\bibliographystyle{apsrev4-1}
\bibliography{Bibliography}
\onecolumngrid

\appendix
\label{app:appendix}
\section{Classical Optimization Method SPSA}
\label{app:appendixSPSA}
Alongside the multi-qubit entangling method the classical optimization method plays an important role in the VQE algorithm, see Fig.~\ref{fig:diagram}. This paper uses the SPSA method which is widely used in VQE research \cite{overview1,Kandala,overview2}\\

SPSA starts with an input initial trial state $|\Psi(\vec{\theta_0})\rangle$ with parameter vector $\vec{\theta_0}\in[0,2\pi]^{\otimes D}$. At every iteration step a small perturbation is done from the parameter vector towards a random direction to obtain two new parameter vectors $\vec{\theta}_{k,\pm}$=$\vec{\theta}_k+c_k\vec{\Delta_k}$. Here $\vec{\Delta_k}$ is a $(3d+2)m$-dimensional random vector that is generated by a sufficiently random generator and $c_k$ is a yet to be defined constant. In this paper $\vec{\Delta_k}$ is determined by the Rademacher distribution where every vector entry has a probability of $1/2$ of being either 1 or -1. The gradient $\vec{g_k}$ is estimated by
\begin{equation}
\label{eq:gradients}
    \vec{g_k}=\frac{\langle \Psi(\vec{\theta_{k,+}})|H|\Psi(\vec{\theta_{k,+}}) \rangle-\langle \Psi(\vec{\theta_{k,-}})|H|\Psi(\vec{\theta_{k,-}}) \rangle}{2c_k}\vec{\Delta_k}.
\end{equation}

\noindent To decrease the energy of the trial state at iteration $k$ a step against the gradient direction is taken such that 
\begin{equation}
    \vec{\theta}_{k+1}=\vec{\theta}_{k}-a_k\vec{g}_k,
\end{equation}

\noindent where $a_k$ is a weight factor dependent on the iteration step. The factors $a_k$ and $c_k$ are given by
\begin{equation}
    a_k=\frac{a}{k^A},\quad c_k=\frac{c}{k^\Gamma}.
\end{equation}
The constants $A, c$ and $\Gamma$ are best chosen to be 0.602, 0.01 and 0.101 respectively to guarantee the smoothest gradient descent \cite{spsacoefficients}.  The constant $a$ is dependent on the gradient around the initial position and is defined according to 
\begin{equation}
    a=\frac{2\pi}{5}\frac{c}{\left\langle|\langle \Psi(\vec{\theta}_{1,+})|H|\Psi(\vec{\theta}_{1,+}) \rangle-\langle \Psi(\vec{\theta}_{1,-})|H|\Psi(\vec{\theta}_{1,-}) \rangle|\right\rangle_{\vec{\Delta}_1}},
\end{equation}

\noindent where $\left\langle...\right\rangle_{\vec{\Delta}_1}$ is the expectation value over the distribution of $\vec{\Delta}_1$. This is the way the average slope around the initialization vector is calculated in order to calibrate the optimization process. In practice this expectation value is approximated using Monte-Carlo methods. A step is taken in the direction opposite of the gradient in order to, conceivably, reach a lower energy value. The SPSA algorithm is a local optimization method.\\

It should be noted that many other classical optimization techniques have been developed and applied to VQE in recent years such as DIRECT\cite{DIRECT1, DIRECT2} and SGD (c.f. Ref.~\cite{overview1} for a comprehensive overview). In particular the parameter shift rule for gradients allows for gradients as in Eq.~\eqref{eq:gradients} to be determined precisely by the QPU \cite{parametergradient}. However, since the focus of this work is on the convergence w.r.t the entanglement methods, we do not discuss further the details of these other methods.\\

Recent developments in VQE research have shown the issue of barren plateaus arising when the number of qubits (and thus the complexity of the problem) is increased \cite{barrenplateau1, barrenplateau2}. These plateaus are large flat parts of the parameter-energy landscape arising in VQE problems which can greatly decrease the effectiveness of a gradient descend method. Since our goal is to show how different entanglement methods and optimizers can influence the convergence of VQE problems even for problems with a small number of qubits this problem is not encountered in the present work. 

\end{document}